\documentclass[12pt,noinfoline]{imsart}

\usepackage{amsmath,amssymb,amsthm,amscd,amsfonts,bm,amsbsy}
\usepackage{graphicx}%
\usepackage{fancyhdr}
\usepackage[left=2.1cm, right=2.1cm,top=2.5cm,bottom=2.1cm]{geometry}
\usepackage{url}
\usepackage{color}
\usepackage[colorlinks,citecolor=blue,urlcolor=blue]{hyperref}

\usepackage{natbib}
\usepackage{wrapfig,verbatim}
\usepackage[normalem]{ulem}
\usepackage{blkarray,enumitem}
\usepackage{tikz-cd,tikz}
\usepackage{mathtools}

\usepackage{algorithm}
\usepackage{algpseudocode}
\algnewcommand\algorithmicinput{\textbf{Input:}}
\algnewcommand\algorithmicoutput{\textbf{Output:}}
\algnewcommand\Input{\item[\algorithmicinput]}%
\algnewcommand\Output{\item[\algorithmicoutput]}%
\usepackage{caption}

\theoremstyle{plain} 
\newtheorem{thm}{Theorem}[section]

\newtheorem{cor}[thm]{Corollary}

\theoremstyle{definition}

\newtheorem{rmk}[thm]{Remark}
\newtheorem{example}[thm]{Example}

\newcommand{\R}{\mathbb{R}}

\newcommand{\Z}{\mathbb{Z}}

\newcommand{\ep}{\epsilon}

\DeclareMathOperator{\Bb}{\mathcal{B}}
\DeclareMathOperator{\Ff}{\mathcal{F}}

\DeclareMathOperator{\Yy}{\mathcal{Y}} 
\DeclareMathOperator{\Oo}{\mathcal{O}} 
\DeclareMathOperator{\Uu}{\mathcal{U}}

\DeclareMathOperator{\Ex}{E} 
\DeclareMathOperator{\pr}{P} 
\DeclareMathOperator{\supp}{supp} 

\newcommand{\fiber}{\Ff_A(u)} 
\DeclareMathOperator{\tr}{\mathsf{T}} 
\DeclareMathOperator{\colsp}{colSpan} 
\DeclareMathOperator{\rank}{rank}
 
\DeclareMathOperator{\dpois}{Poisson}
\DeclareMathOperator{\dgamma}{Gamma}
\DeclareMathOperator{\dunif}{Uniform}
\DeclareMathOperator{\Andbf}{\textbf{ and }}

\definecolor{forestgreen}{rgb}{0,.72,0} 
\definecolor{brickred}{rgb}{.72,0,0}
\definecolor{darkcerulean}{rgb}{0.03, 0.27, 0.49}

\begin{document}

\title{Sampling lattice points in a polytope: a Bayesian biased algorithm with random updates 
}

\author{\fnms{Miles Bakenhus}\ead[label=e1]{mbakenhus@hawk.iit.edu}},
\author{\fnms{Sonja Petrovi\'c}\ead[label=e2]{sonja.petrovic@iit.edu}}

\date{\today}

\begin{abstract}
The set of nonnegative integer lattice points in a polytope, also known as the fiber of a linear map, makes an appearance in several applications including optimization and statistics. 
We address the problem of sampling from this set using three ingredients: an easy-to-compute lattice basis of the constraint matrix, a biased sampling algorithm with a Bayesian framework, and a step-wise selection method. The bias embedded in our algorithm updates sampler parameters to improve fiber discovery rate at each step chosen from previously discovered elements.
We showcase the performance of the algorithm on several examples, including fibers that are out of reach for the state-of-the-art Markov bases samplers. 
\end{abstract} 

\maketitle

%
%
%
%

%


Fix an integer matrix  $A \in \Z^{N\times M}$ and a vector $u\in \Z^{N}$, such that the system $Ax=u$ has a solution $x_0\in\Z_{\geq 0}^{M}$. Consider the feasible polytope $P:=\{x\in\R^N:Ax=u,x\geq0\}$ of all nonegative solutions to the linear system. 
The set of  integer lattice points  $P\cap \mathbb Z_{\geq 0}^M$ is called the \emph{fiber} of $u$ under the model $A$. We denote it as follows: 
	\[\fiber := \{x\in \Z_{\ge 0}^{M} \colon Ax = u\}.\] 
In this paper, we construct a new fiber sampling algorithm called  Random Updating Moving Bayesian Algorithm (RUMBA).  The input is the constraint matrix $A$, the vector $u$, and one feasible point $x_0$, and the output is a sampled subset of the $\fiber$ or, if RUMBA runs sufficiently long,  the entire set of nonnegative integer points in $P$. 
Relying on the fact that a difference of two points $x,y\in\fiber$ lies in the lattice $\ker_\mathbb Z A$, the  sampler computes a vector space basis of the lattice $\ker_\mathbb Z A$ and then takes random linear combinations of these basis elements to discover new points in the fiber. The size of the combinations and coefficients are drawn from some, typically flat, conjugate prior distribution. After a user-determined number of samples the distribution is updated and the posterior is used in the next iteration. 
For fixed runtime parameters, given the matrix $A$ and a lattice basis of $\ker_{\Z} A$, our algorithm runs in $\Oo(MK)$ time where $K$ is the number of moves in the basis. When the basis is a minimal spanning set of moves, $K = M-\rank(A)$ such that the algorithm runs in quadratic time with respect to the number of columns of $A$. The rate of discovery for the fiber is related to the sparsity of linear combinations of basis moves needed to connect the fiber, as well as the polytope diameter with respect to these moves. 

Integer points in polytopes  make a fundamental appearance in several applications, including optimization and statistics. 
In discrete optimization, the set $\fiber$ is the support set of an integer program $\min_{x\in\mathbb Z^M : Ax=u,  x\geq 0} f(x)$  for optimizing a linear objective function $f(x)=c^T x$ over the fiber; \citep{DeLoeraHemmeckeBook} is an excellent resource. 
In statistics, the fiber $\fiber$ is the support of the conditional distribution given the value of the sufficient statistic $u$ of an observed data $x$, under the so-called log-affine \citep{Lauritzen} or \emph{log-linear model} defined by the matrix $A$.
 \cite{BesagClifford89}  is an early example of applications of fiber sampling fundamental to statistical inference---although they do not refer to the set $\fiber$ as a `fiber'---and demonstrates the practicality and need for theoretical advances that can help devise irreducible Markov chains on fibers. 
In the 1990s, \cite{DS98} introduced a new sampling algorithm from $\fiber$, showing  how to explicitly construct such a Markov chain using the Metropolis-Hastings algorithm for any integer matrix $A$ and any $u$, using moves that are constructed from binomials using nonlinear algebra.  They name any finite collection of moves resulting in an aperiodic, irreducible chain a \emph{Markov basis} and prove that  any set of moves corresponding to a generating set of the toric ideal $I_A$ suffices.
\cite{petrovic2019} gives an overview of Markov bases for a general audience, while   \cite{MarkovBases25years} contains a literature overview of Markov bases from a practical point of view.  In algebraic statistics, log-linear models are sometimes called \emph{toric models}, due to their connection to toric varieties. 
  \cite{GeigerMeekSturmfels} put discrete graphical models in the context of toric  models, extending the reach of the fundamental theorem by \citeauthor{DS98}. 
Recently, statistical applications of fiber sampling to the key question of model/data fit were extended beyond log-linear 
 models; for example, \cite{karwa2016exact} extend the use of Markov bases and sampling to \emph{mixtures} of log-linear models, specifically in the context of latent-variable random networks. See Remark~\ref{remark:must cite!} for broader related literature. 

\medskip 
We are interested in sampling from $\fiber$ using an efficient algorithm which \emph{adapts as it discovers new points} in $\fiber$.  The efficiency comes from the fact that we do not use non-linear algebra, meaning we do not compute a Markov basis, relying instead  on the following key observation: 
the only bases of the lattice $\ker_\mathbb ZA$ that are computable by linear algebra are \emph{lattice bases}. 
A lattice basis consists of any set of vectors that span the integer kernel $\ker_\mathbb Z A$ as a vector space, and the minimal such set has co-rank many elements. Unfortunately, as these vectors generally correspond to a strictly smaller set of binomials than a generating set of the toric ideal $I_A$ (see Section~\ref{section:bases}), they do not constitute a Markov basis, and as such they cannot be used to construct a connected Markov chain directly.  However, a finite  integer linear combination of lattice bases elements will produce a Markov basis; the issue is that it is not known how large the combination should be, a priori, or what is an optimal set of coefficients in the combination to discover the fiber at a high rate as one samples. 
   \citeauthor{DS98} report that they \emph{``tried this idea [of random combinations of lattice basis elements] in half a dozen problems and found it does not work well"}, in the sense that on some examples, a random walk on the fiber using combinations of lattice basis elements takes millions of steps to converge, compared to a only few hundred steps with a Markov basis. \label{lattice bases do not work}
A decade later, the same idea appears in \cite{HAT12}, where \emph{``with many examples [the authors] show that the approach with lattice bases is practical"} for several statistical models. The key idea is to use a \emph{random} combination of lattice bases moves, where the coefficients are selected from a Poisson distribution, or any other distribution with infinite support, so as to guarantee a connected chain. Both papers report that choosing the Poisson parameter is \emph{``a delicate operation"} and no trivial matter. 

With this in mind, we take the next statistically logical step:  embed the idea of constructing random linear combinations with draws from a Poisson distribution into a Bayesian framework. 
Without knowing or computing the Markov basis, the RUMBA algorithm samples the points in the neighborhood of the current fiber point, and then---and this is crucial---learns which directions on the lattice are more likely to discover more points. 
 In this sense, the algorithm begins in a model-and-fiber-agnostic way, but then adjusts its parameters according to the fiber it is sampling. The intuition behind introducing bias  in this way lies in the convexity of the polytope: from a given starting point $x_0$, some directions on the lattice will travel toward the interior of the fiber, while others will result in a point outside. By convexity, there is no reason for exploring the directions going outside starting at the same $x_0$. Instead, it will be much better to bias the sampler in directions that are more likely to lead to new points in $\fiber$. 

Since our sampler  proceeds in several stages, the following is  a high-level overview to orient the reader.  
Steps 1-3 correspond to the SAMPLE  Algorithm~\ref{alg:bfs}. Step 4 is the  UPDATE Algorithm~\ref{alg:ibfs}, and Step 5 is the RUMBA Algorithm~\ref{alg:mibfs}. Figure~\ref{fig:itrParams} is a schematic of how biasing the sampler in Step 2 below moves toward discovering more of the fiber points. 
\begin{enumerate}\label{algo outline}
	\item Generate a batch of samples $X_1,\dots,X_J$ from the distribution $p(x\mid x_0,\theta)$, for some choice parameters $\theta\in \Theta$, such that the support of $p(x\mid x_0,\theta)$ contains all solutions of $Ax=u$ for any  $\theta\in\Theta$.
	\item For sampled $X_j$ that are solutions to $Ax=u$, update the parameters $\theta$ to bias the next batch of samples in towards previously sampled nonnegative solutions to the equation.
	\item Repeat this process of sampling batches and updating parameters for a specified number of iterations. 
	\item Select a new initial solution $x_1$ from the previously sampled solutions, and reset the parameters $\theta$ to their initial state. Then repeat the batch sampling and parameter updates using these values such that new samples are distributed: $X_j\sim p(x\mid x_1,\theta)$.
	\item Repeat all of the previous steps, updating the initial solution $x_t$ after completing each set of batch samples for $t = 1,\dots, T$ where $T\in \Z_{>0}$. 
\end{enumerate}

\begin{figure}[htb]
	\begin{tikzpicture}	[scale=0.75, auto]
		
		\tikzstyle{point}=[inner sep=1pt, fill=black, draw=black, shape=circle]
		\tikzstyle{opoint}=[inner sep=3pt, fill=red,fill opacity=0.2, draw=black, shape=circle]
		\tikzstyle{none}=[inner sep=0pt,fill=none, draw=none]
		\tikzstyle{edge}=[-, fill=none]
		\tikzstyle{dashed}=[densely dashed, thick, fill=none]
		\node [style=point] (4) at (0, 3) {};
		\node [style=point] (6) at (1, 2) {};
		\node [style=point] (7) at (1, 3) {}; 
		\node [style=point] (8) at (2, 2) {};
		\node [style=point] (9) at (2, 3) {};
		\node [style=point] (10) at (3, 2) {};
		\node [style=point] (11) at (3, 3) {};
		\node [style=point] (12) at (4, 2) {};
		\node [style=point] (13) at (4, 3) {};
		\node [style=point] (14) at (5, 2) {};
		\node [style=point] (15) at (6, 2) {};
		\node [style=point] (16) at (6, 3) {};
		\node [style=point] (17) at (5, 3) {};
		\node [style=point] (19) at (5, 4) {};
		\node [style=point] (20) at (4, 4) {};
		\node [style=point] (21) at (3, 4) {};
		\node [style=point] (22) at (2, 4) {};
		\node [style=point] (23) at (1, 4) {};
		\node [style=point] (24) at (0, 4) {};
		\node [style=point] (25) at (0, 5) {};
		\node [style=point] (26) at (1, 5) {};
		\node [style=point] (27) at (2, 5) {};
		\node [style=point] (28) at (3, 5) {};
		\node [style=point] (29) at (4, 5) {};
		\node [style=point] (30) at (5, 5) {};
		\node [style=point] (31) at (6, 5) {};
		\node [style=point] (32) at (0, 2) {};
		\node [style=point] (33) at (6, 4) {};
		
		\node [style=none] (38) at (2.821, 2.756) {}; 
		\node [style=none] (39) at (2.520, 4.032) {};  
		\node [style=none] (40) at (-0.821, 3.244) {};  
		\node [style=none] (41) at (-0.520, 1.968) {}; 
		
		\node [style=opoint] (e0) at (1, 3) {}; 
		
		\coordinate[label ={[label distance=3pt]270:{\tiny$\qquad\Ex_0[X]=x_0$}}] (x0) at (7);
		
		\draw [style=dashed, bend right=45] (38.center) to (39.center);  
		\draw [style=dashed, bend right=45] (39.center) to (40.center);  
		\draw [style=dashed, bend right=45] (40.center) to (41.center);  
		\draw [style=dashed, bend left=45] (38.center) to (41.center);  
		
		\node [style=none] (pad) at (1, 0.75) {};  
		
		\node [style=none] (34) at (0, 2.5) {};
		\node [style=none] (35) at (6.25, 5.75) {};
		\node [style=none] (36) at (5.5, 2) {};
		\node [style=none] (37) at (1.5, 2) {};
		\draw [style=edge] (34.center) to (35.center);
		\draw [style=edge] (35.center) to (36.center);
		\draw [style=edge] (36.center) to (37.center);
		\draw [style=edge] (34.center) to (37.center);
		
	\end{tikzpicture}
	$\qquad$\begin{tikzpicture}	[scale=0.75, auto]
		\tikzstyle{point}=[inner sep=1pt, fill=black, draw=black, shape=circle]
		\tikzstyle{opoint}=[inner sep=3pt, fill=red,fill opacity=0.2, draw=black, shape=circle]
		\tikzstyle{none}=[inner sep=0pt,fill=none, draw=none]
		\tikzstyle{edge}=[-, fill=none]
		\tikzstyle{dashed}=[densely dashed, thick, fill=none]
		
		\node [style=point] (4) at (0, 3) {};
		\node [style=point] (6) at (1, 2) {};
		\node [style=point] (7) at (1, 3) {}; 
		\node [style=point] (8) at (2, 2) {};
		\node [style=point] (9) at (2, 3) {};
		\node [style=point] (10) at (3, 2) {};
		\node [style=point] (11) at (3, 3) {};
		\node [style=point] (12) at (4, 2) {};
		\node [style=point] (13) at (4, 3) {};
		\node [style=point] (14) at (5, 2) {};
		\node [style=point] (15) at (6, 2) {};
		\node [style=point] (16) at (6, 3) {};
		\node [style=point] (17) at (5, 3) {};
		\node [style=point] (19) at (5, 4) {};
		\node [style=point] (20) at (4, 4) {};
		\node [style=point] (21) at (3, 4) {};
		\node [style=point] (22) at (2, 4) {};
		\node [style=point] (23) at (1, 4) {};
		\node [style=point] (24) at (0, 4) {};
		\node [style=point] (25) at (0, 5) {};
		\node [style=point] (26) at (1, 5) {};
		\node [style=point] (27) at (2, 5) {};
		\node [style=point] (28) at (3, 5) {};
		\node [style=point] (29) at (4, 5) {};
		\node [style=point] (30) at (5, 5) {};
		\node [style=point] (31) at (6, 5) {};
		\node [style=point] (32) at (0, 2) {};
		\node [style=point] (33) at (6, 4) {};

		\node [style=none] (e0) at (1, 3) {};  
		\node [style=opoint](e1) at (2.450, 3.250) {};  
		
		\coordinate[label =below:{\tiny$x_0$}] (x0) at (7);
		\coordinate[label ={[label distance=1pt]20:{\tiny$\Ex_1[X]$}}] (x1) at (e1);
		
		\node [style=none] (38) at (4.173, 2.987) {}; 
		\node [style=none] (39) at (3.936, 4.161) {};  
		\node [style=none] (40) at (0.727, 3.513) {};  
		\node [style=none] (41) at (0.964, 2.339) {}; 
		
		\draw [style=dashed, bend right=45] (38.center) to (39.center);  
		\draw [style=dashed, bend right=45] (39.center) to (40.center);  
		\draw [style=dashed, bend right=45] (40.center) to (41.center);  
		\draw [style=dashed, bend left=45] (38.center) to (41.center);  

		\node [style=none] (pad) at (1, 0.75) {};  
		\node [style=none] (34) at (0, 2.5) {};
		\node [style=none] (35) at (6.25, 5.75) {};
		\node [style=none] (36) at (5.5, 2) {};
		\node [style=none] (37) at (1.5, 2) {};
		\draw [style=edge] (34.center) to (35.center);
		\draw [style=edge] (35.center) to (36.center);
		\draw [style=edge] (36.center) to (37.center);
		\draw [style=edge] (34.center) to (37.center);
		
	\end{tikzpicture}
	$\qquad$\begin{tikzpicture}	[scale=0.75, auto]
		\tikzstyle{point}=[inner sep=1pt, fill=black, draw=black, shape=circle]
		\tikzstyle{opoint}=[inner sep=3pt, fill=red,fill opacity=0.2, draw=black, shape=circle]
		\tikzstyle{none}=[inner sep=0pt,fill=none, draw=none]
		\tikzstyle{edge}=[-, fill=none]
		\tikzstyle{dashed}=[densely dashed, thick, fill=none]
		
		\node [style=point] (4) at (0, 3) {};
		\node [style=point] (6) at (1, 2) {};
		\node [style=point] (7) at (1, 3) {}; 
		\node [style=point] (8) at (2, 2) {};
		\node [style=point] (9) at (2, 3) {};
		\node [style=point] (10) at (3, 2) {};
		\node [style=point] (11) at (3, 3) {};
		\node [style=point] (12) at (4, 2) {};
		\node [style=point] (13) at (4, 3) {};
		\node [style=point] (14) at (5, 2) {};
		\node [style=point] (15) at (6, 2) {};
		\node [style=point] (16) at (6, 3) {};
		\node [style=point] (17) at (5, 3) {};
		\node [style=point] (19) at (5, 4) {};
		\node [style=point] (20) at (4, 4) {};
		\node [style=point] (21) at (3, 4) {};
		\node [style=point] (22) at (2, 4) {};
		\node [style=point] (23) at (1, 4) {};
		\node [style=point] (24) at (0, 4) {};
		\node [style=point] (25) at (0, 5) {};
		\node [style=point] (26) at (1, 5) {};
		\node [style=point] (27) at (2, 5) {};
		\node [style=point] (28) at (3, 5) {};
		\node [style=point] (29) at (4, 5) {};
		\node [style=point] (30) at (5, 5) {};
		\node [style=point] (31) at (6, 5) {};
		\node [style=point] (32) at (0, 2) {};
		\node [style=point] (33) at (6, 4) {};
		
		\node [style=none] (e0) at (1, 3) {};  
		\node [style=opoint](e2) at (3.870, 3.130) {};  
		
		\coordinate[label =below:{\tiny$x_0$}] (x0) at (7);
		\coordinate[label ={[label distance=2pt]160:{\tiny$\Ex_2[X]$}}] (x1) at (e2);
		\node [style=none] (38) at (6.307, 2.829) {}; 
		\node [style=none] (39) at (5.871, 4.553) {};  
		\node [style=none] (40) at (1.433, 3.431) {};  
		\node [style=none] (41) at (1.869, 1.707) {}; 
		
		\draw [style=dashed, bend right=45] (38.center) to (39.center);  
		\draw [style=dashed, bend right=45] (39.center) to (40.center);  
		\draw [style=dashed, bend right=45] (40.center) to (41.center);  
		\draw [style=dashed, bend left=45] (38.center) to (41.center);  
		\node [style=none] (pad) at (1, 0.75) {};  
		\node [style=none] (34) at (0, 2.5) {};
		\node [style=none] (35) at (6.25, 5.75) {};
		\node [style=none] (36) at (5.5, 2) {};
		\node [style=none] (37) at (1.5, 2) {};
		\draw [style=edge] (34.center) to (35.center);
		\draw [style=edge] (35.center) to (36.center);
		\draw [style=edge] (36.center) to (37.center);
		\draw [style=edge] (34.center) to (37.center);
		
	\end{tikzpicture}
	\caption{A polytope $Ax=u$ is shown in each of the figures above together with points in $\Z^2$. \textbf{Left:} The initial sample expectation $\Ex_0[X]=x_0$ is circled in red, and the dotted ellipse has semi-major and minor axes constructed from the eigenvalues and vectors of the covariance matrix of $X$. \textbf{Middle:} The expectation, $\Ex_1[X]$, and covariance after 1 parameter update. \textbf{Right:} The expectation, $\Ex_2[X]$, and covariance after 2 parameter updates. }\label{fig:itrParams}
\end{figure}
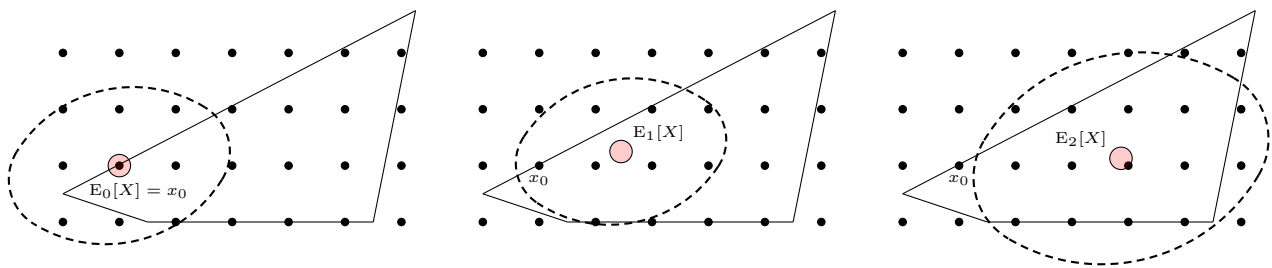

The manuscript is organized as follows. Theoretical background on the many bases of an integer lattice is in Section~\ref{section:bases}. The algorithm  is presented in Section~\ref{section:algorithms}, and several illustrative examples in Section~\ref{section:simulations}.  We close with with a discussion on parameter tuning in Section~\ref{section:practical}.

The way we run RUMBA in this paper is with the goal of discovering fiber points at a high rate. In particular, this means that our runtime parameters are set to discover more new points. If the goal is different, say, to construct binomials in the toric ideal $I_A$ which is of interest, for example, in \cite[Section 3.3]{gbViolator}, then one should bias toward those combination coefficients that give new \emph{directions} for moving about the lattice. Section~\ref{section: runtime parameters} further reflects on parameter choices. 
In particular, it would be of interest to also explore the following problem: how to bias the RUMBA sampler toward producing a certain kind of moves. The most comprehensive way of doing so is to put a conjugate prior distribution on the  probability of picking each move in the lattice basis. It is an open problem to determine how to do this to, for example, produce indispensable binomials in the ideal $I_A$, and more theoretical results in that direction are needed. For motivation, see \cite{CTV-indispensable}. 

\section{Tradeoff between basis complexity and fiber connectivity} \label{section:bases}

 The fiber sampling problem is  inherently a difficult one, if for no other reason than for the sheer size of the problem:  \cite{latte,latte2} provide a polynomial time algorithm to compute the size of $\fiber$, and even for small matrices $A$ it can be quite large. 
 In this section, we provide context for the main reasons why sampling algorithms on fibers may `get stuck'. 
A Markov basis for a fixed matrix $A$ guarantees that the Markov chain constructed from it is irreducible for every value of $u$. This means that \emph{one} Markov basis suffices to connect \emph{all} fibers, each of which is a  translate of $\ker_\mathbb Z A$. Naturally, there exist instances in which the performance of the random walk is not optimal; this is well-documented in the (algebraic) statistics literature; see, for example,  \cite{SteveAleMe-holland}, \cite[Problem 5.5]{DobraEtAl-IMA}: ``Markov bases are data-agnostic";  and \cite{MarkovBases25years} for context. 

Given a Markov basis $\mathcal{M}$ for  an integer matrix $A$, the \citeauthor{DS98} chain proceeds by constructing a Metropolis-Hastings algorithm on the fiber as a random walk on $\fiber$ from a given starting point $x_0\in\fiber$.  One step from $x_0$ consists of uniformly picking an element $m\in \mathcal{M}$ and choosing $\ep\in \{\pm 1\}$ with probability $\frac{1}{2}$ independently of $m$. The chain then moves to $x_0+\ep m$ if $x_0+\ep m\geq 0$ and stays at $x_0$ otherwise. 
Of course, in statistics, the algorithm includes the acceptance ratio which controls the stationary distribution of the chain,  see  \cite[Algorithm 1.1.13]{DSS09} for details.  
Here we do not discuss the Hastings ratio, because rather than focusing on converging to a particular distribution, our goal is to discover the fiber as quickly as possible. 

The problem of Markov bases being data-agnostic amounts  to two facts. One is that most of the moves are not needed for one given fiber because they will result in negative entries in $x_0+\ep m$, but the user does not know which moves will do so a priori, and discarding some moves may reduce the Markov chain.  This is compounded by the second issue of low connectivity of the underlying fiber graph, defined as follows. Let $M\subset \ker_\mathbb Z A$ be any set of moves on the fiber. Of course, $x-y\in\ker_\mathbb ZA$ holds for any two points $x,y\in\fiber$.  The \emph{fiber graph} $G(\fiber,M)$ is a graph with vertex set $\fiber$ such that $x,y\in\fiber$ are connected if and only if $x-y\in M$ or $y-x\in M$. 
The choice of a set $M$ will change the fiber graph; in particular, by definition, a  \emph{Markov basis} is any set $M$ that is (minimally) necessary for connectivity of the  graphs $G(\fiber,M)$ for all $u$. 
Working with a \emph{lattice basis} $B$ of the matrix $A$, as we do in the SAMPLE step, does \emph{not} produce a connected fiber graph. This is the reason why RUMBA constructs random combinations of lattice basis elements, producing, with nonnegative probability, every possible edge in on the fiber graph. 
A thorough case study of the different bases can be found in \cite{DSS09}, Section 1.3, which has the title ``The many bases of an integer lattice".  
Unfortunately, computing a Markov basis requires elimination and Gr\"obner bases (see \cite{CLO}, \cite[Chapter 4]{St}). Even so, a minimal Markov basis results in a fiber graph with far less edges than other larger bases of the matrix $A$. Figure~\ref{fig:3 fiber graphs} illustrates how the graph changes with the basis. 
\begin{figure}[h]
\includegraphics[scale=.7]{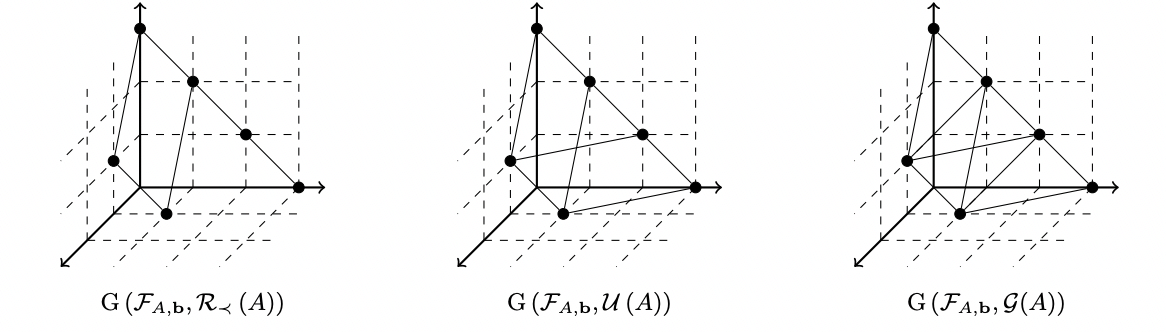}
\caption{A figure from \cite{HemmeckeWindisch:FiberGraphs}, showing three different fiber graphs for the same underlying fiber. The edges on the left figure represent moves from a reduced Gr\"obner basis, in the middle are moves from the universal Gr\"obner basis, and on the right is the Graver basis. 
\label{fig:3 fiber graphs}
}
\end{figure}

There are families of examples of `bad fibers' which are difficult to connect with Markov or Graver moves not because the moves are insufficient, but because the right-hand-side $u$ of the equation $Ax=u$ is such that it forces the solutions $x$ to live in spaces which are connected only by one or a small subset of the moves. A typical family of examples are presented in \cite[Sections 4 and 5]{HemmeckeWindisch:FiberGraphs}; we extract Figure~\ref{figure: bad fibers} which shows a bottleneck edge for fiber connectivity. 
 This type of a fiber is well-known in algebraic statistics; see also a discussion of indispensable binomials in \cite{AokiIndispensable}. 
\begin{figure}[h]
\includegraphics[scale=.6]{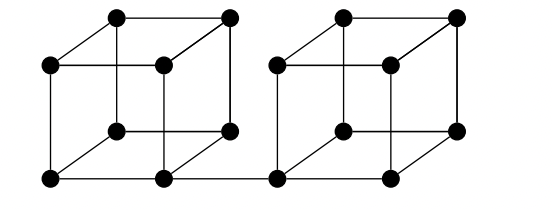}
\caption{A figure from \cite{HemmeckeWindisch:FiberGraphs} depicting a fiber with a bottleneck edge which  will have a low probability of being selected by any algorithm that randomly selects from a predetermined set of moves.}\label{figure: bad fibers}
\end{figure}
In Section~\ref{section: HW fibers} we will show the performance of RUMBA on this family of fibers.

In an ideal world, one should strive to construct a complete fiber graph because, in theory, if  one samples edges(moves) from a complete graph uniformly, one gets maximum conductance and therefore rapidly mixing Markov chains on fibers.  \citeauthor{Tobias2015mixing} makes a very intuitive suggestion: \emph{`A possible way out is to adapt the Markov basis appropriately so that its complexity grows with the size of the right-hand entries. This can be achieved by adding a varying number of $\mathbb Z$-linear combinations of the moves in a way that the edge-expansion of the resulting graph can be controlled.'} 
The issue we---and anyone constructing an algorithm meant to work for all fibers---face is that we do not know the fiber graph, we do not see the lattice structure, and we do not know how many is the `varying number' of moves, a priori. This is precisely the motivation for letting RUMBA adjust parameters and use bias to `learn on the go'. 
In fact, the RUMBA sampler does sample edges from the complete fiber graph, however the choice of moves or edges is not uniform, simply because it is computationally intractable to compute all possible moves on fibers of exponential size. 

\begin{rmk}[Sampling constraints]
By definition, $\fiber$ is  unconstrained above  except by the total sample size, that is, the 1-norm of $x_0$. 
In some applications, one may need to sample from a \emph{restricted} fiber: perhaps some entries of $x$ are set to $0$, or $x\in\{0,1\}^N$ rather than $x\in\mathbb Z^N$. 
It is known that restricted Graver bases suffice for connecting such fiber subsets \cite[Proposition 3.3]{MarkovBases25years}, but of course they are difficult to compute for large matrices $A$. 
There is another body of literature on sampling $0/1$ fibers in the context of graph algorithms; \cite{MixingTimeSwitchChain-overview} provides an excellent summary of the state of the art on rapidly mixing Markov chains in this context.  
\end{rmk} 

\begin{rmk}[Broader related literature]\label{remark:must cite!}
One could abandon pure MCMC methods entirely and devise alternative sampling algorithms. One such success story is \cite{KahleYoshidaGarciaPuente} who, motivated by the high complexity of determining a Markov basis for some models, combine it with sequential importance sampling (SIS). Prior work on purely using SIS, however, was less impressive, as \cite{Dob2012} found, in numerical comparisons, that the a Markov bases approach computed dynamically performed better than SIS. 
\citeauthor{KahleYoshidaGarciaPuente} state that the motivation behind trying to use moves as large as possible in order for the chain to `get random much more rapidly' stems from the hit and run algorithm, and relate this to random Poisson linear combinations of lattice bases elements form \cite{HAT12}, which is discussed in more detail in \cite[Chapter 16]{AHT2012}. 
For relevant background on Markov bases and the connection between statistics and nonlinear algebra, we  refer the interested reader to one of the algebraic statistics textbooks \cite{SethBook,DSS09}. 
 A notable related body of work is summarized in \cite{Diaconis22approximateExchangability}, where Markov bases are key to formulating partial exchangeability for contingency tables; see also \cite{Diaconis2022exchangabilityTables}. 
\end{rmk} 

\section{RUMBA: Random Updating Moving Bayesian Algorithm} \label{section:algorithms}

Since the algorithm uses  layers of iterations---sampling, iterate,s and time steps---and vectors of parameters, for the convenience of the reader, we have collected all of the notation in one place in Appendix~\ref{section:appendix}. For clarity, all of the indices, matrices, vectors, and random variables are also defined in the main body of the text in this section. 




Given a matrix $A \in \Z^{N \times M}$ and a vector $u\in \Z^{N}$, denote the \textit{fiber} of $A$ as before: 
$\Ff_A(u) =  \{x\in \Z_{\ge 0}^{M} \colon Ax = u\}$.
To ease the notation in this section, since the matrix $A$ and the vector $u$ are fixed, we will drop $A$ and $u$ from the fiber notation: 
\[\Ff := \fiber.\]
For some basis $\Bb = \{b_1,\dots,b_{K}\}$ of $\ker_{\Z} A$, define the matrix
\[B = \begin{pmatrix}
	b_1 & \dots & b_{K}
\end{pmatrix}.\]
For any distinct fiber elements $x,x^\prime \in \Ff$, $Ax = Ax^\prime$ such that  $x-x^\prime \in \ker_{\Z}(A)$. Then $x-x^\prime = By$ for some vector $y\in \Z^{K}$. In other words, given some initial $x_0 \in \Ff$, for any $x \in \Ff$ there is a coefficient vector $y\in \Z^{K}$ such that 
\[x = x_0 + By.\]
Algorithm \ref{alg:bfs} attempts to sample elements of $\Ff$ using a mixture of independently sampled Poisson random variables to generate coefficient vectors $y$ in the equation $x = x_t + By$ for some known fiber element $x_t$. For each step $t\in [T]$ each $x_t$ is chosen via some selection method. For example, $x_t$ may be sampled from a mixture of the uniform distribution of all previously sampled fiber elements and a uniform distribution of only the newly sampled elements in step $t-1$. 

During the $i^{\text{th}}$ iteration of the $t^{\text{th}}$ step, for a given $x_t \in \Ff$, the $k^\text{th}$ element in the $j^{\text{th}}$ sampled coefficient vector $Y_{i,j,k}$ is obtained by sampling a mixture of the Poisson random variables
\[Y^+_{i,j,k} \sim \dpois(\lambda^+_{i,k}),\quad Y^-_{i,j,k} \sim \dpois(\lambda^-_{i,k}), \quad \text{for } k\in [K]\]
for parameters $\lambda^+_{i,k},\lambda^-_{i,k} \ge 0$, such that $Y_{i,j,k}= Y^+_{i,j,k} - Y^-_{i,j,k}$. This produces the following vector in $\Z^M$ to check as a potential fiber element:
\[X_{t,i,j} = x_t + BY_{i,j}.\]
Since $BY_{i,j}$ is a linear combination of vectors in the kernel of $A$,  $AX_{t,i,j} = A(x_t + BY_{i,j}) = u$. Therefore, $X_{t,i,j} \in \Ff$ whenever $X_{t,i,j} \ge 0$. 

Parameters of the Poisson random variables are updated at each iteration in Algorithm \ref{alg:ibfs} by using the values of the $Y^+_{i,j,k}$, and $Y^-_{i,j,k}$ corresponding only to the samples $X_{t,i,j} \ge 0$ where $X_{t,i,j}$ has not previously been sampled. Denote newly sampled fiber elements in iteration $i$ for initial $x_t$ by $\Ff_{t,i}^*\subseteq \Ff_{t,i}$ such that $\Ff_{t,i}\backslash \Ff^*_{t,i} = \Ff_{t,i-1}$, and denote the corresponding coefficients 
\[\Yy_{t,i}^* = \{Y_{i,j} \in \Yy_{t,i} \colon \exists j \in [J], \text{ such that } X_{t,i,j}\in \Ff_{t,i}^* \text{ and } X_{t,i,j} = x_t + BY_{i,j} \},\]
where $\Yy_{t,i}$ is the set of all sampled coefficients that yield a fiber element up to the $i^{\text{th}}$ iteration. The following Bayesian prior for these parameters is used:
\[\lambda^+_{i,k} \sim \dgamma(\alpha^+_{i-1,k},\beta^+_{i-1,k}),\quad\lambda^-_{i,k}\sim \dgamma(\alpha^-_{i-1,k},\beta^-_{i-1,k}), \quad\text{for } k\in [K], \]
where $\alpha^+_{i-1,k}, \alpha^-_{i-1,k}$'s are shape parameters and $\beta^+_{i-1,k}, \beta^-_{i-1,k}$'s are the rate parameters. This is a conjugate prior with the following posterior distribution given elements of $\Ff_{t,i}^*$ and the corresponding sampled coefficients $\Yy_{t,i}^*$ in the $i^{\text{th}}$ iteration,
\[p(\lambda^{\pm}_{i+1,k}\mid \Ff_{t,i}^*,\Yy_{t,i}^*) \sim \dgamma\left(\alpha_{i,k},\, \beta_{i,k} \right),\]
where 
\[\alpha_{i,k}^{\pm} =\alpha^\pm_{i-1,k} +\sum_{Y_{i,j}\in \Yy_{t,i}^*}Y^{\pm}_{i,j,k}, \quad   \beta_{i,k}=\beta^{\pm}_{i-1,k} + |\Ff_{t,i}^*|.\]
This implies that the expected value of $\lambda_{i+1,k}^{\pm}$ given $\Ff^*_{t,i}$ and $\Yy_{t,i}^*$ is
\[\Ex [\lambda^{\pm}_{i+1,k}| \Ff_{t,i}^*,\Yy_{t,i}^*] = \frac{\alpha^\pm_{i,k}}{\beta^{\pm}_{i,k}}.\]
This expectation is the reason why, in Algorithm \ref{alg:bfs}, we set $\lambda^{\pm}_{i,k} = \Ex [\lambda^{\pm}_{i,k}| \Ff_{t,i-1}^*,\Yy_{t,i-1}^*]$. However, passing the updated shape and rate parameters through each iteration allows us to sample $\lambda^{\pm}_{i,k}$ if desired. 


\begin{algorithm}[!t]
	\captionsetup{font=small}
	\caption{Fiber Sample Loop (\textsc{Sample})}\label{alg:bfs} 
	\begin{algorithmic}[1]
		\Input $x_t\in \Ff,\,J\in\Z_{>0},\,B,\,\alpha^+_{i-1},\,\alpha^-_{i-1},\,\beta^+_{i-1},\,\beta^-_{i-1}\,\Ff_{t,i-1}$
		\Output $\alpha^+_{i},\,\alpha^-_{i},\,\beta^+_{i},\,\beta^-_{i},\,\Ff_{t,i}$
		\State $K = \#(\text{columns of } B)$
		\State $\Ff_{t,i} = \Ff_{t,i-1}$
		\State $\lambda^+_i = \alpha^+_{i-1} / \beta^+_{i-1}$ \Comment{Elementwise division}
		\State $\lambda^-_{i} = \alpha^-_{i-1} / \beta^-_{i-1}$
		\State $\alpha^+_{i} = \alpha^+_{i-1}$
		\State $\beta^+_{i} = \beta^+_{i-1}$
		\State $\alpha^-_{i} = \alpha^-_{i-1}$
		\State $\beta^-_{i} = \beta^-_{i-1}$
		\For{$j\in [J]$}
		\For{$k \in [K]$}
		\State $Y^{+}_{i,j,k} \sim \dpois(\lambda^+_{i,k})$
		\State $Y^{-}_{i,j,k} \sim \dpois(\lambda^-_{i,k})$
		\EndFor
		\State $Y_{i,j} = Y^{+}_{i,j} - Y^{-}_{i,j}$ 
		\State $x = x_t + BY_{i,j}$
		\If{$x \ge 0 \Andbf x\notin \Ff_{t,i}$}
		\State $\Ff_{t,i} = \Ff_{t,i} \cup \{x\}$
		\State $\alpha^+_i =\alpha^+_i+ Y^+_{i,j} $ \Comment{sum of feasible $Y^{+}_{i,j}$'s}
		\State $\alpha^-_i =\alpha^-_i+ Y^-_{i,j} $ \Comment{sum of feasible $Y^{-}_{i,j}$'s}
		\State $\beta^+_i = \beta^+_i + 1$ \Comment{$\beta^{\pm}_{i+1}$'s are the total number of feasible samples}
		\State $\beta^-_i = \beta^-_i + 1$					
		\EndIf
		\EndFor					
	\end{algorithmic}
\end{algorithm}

Algorithm \ref{alg:ibfs} iterates Algorithm \ref{alg:bfs} for a given starting element from a known sample of the fiber. In doing so, the parameters $\lambda_{i,k}^{\pm}$'s shift towards some centroid for $Y_{i,j}$ giving the expectation:
\begin{align*}
\Ex\left[x\mid x_t,\Ff_{t,i-1},\Ff_{t-1},\Yy_{t,i-1}\right] & = x_t + B(\lambda_i^{+}-\lambda_i^-) \\
& = x_t + B\left(\frac{\alpha^+_{i-1,k}}{\beta^{+}_{i-1,k}}-\frac{\alpha^-_{i-1,k}}{\beta^{-}_{i-1,k}}\right)\\
& = x_t + B\left[\frac{1}{1+|\Ff_{t,i-1}\backslash \Ff_{t-1}|}\left((\alpha_0^+ - \alpha_0^-)+\sum_{Y_{r,j}\in \Yy_{t,i-1}}Y_{r,j}\right)\right].
\end{align*}
When $\alpha^{+}_0 = \alpha_0^-$ and $\beta^{\pm}_0 = \boldmath{1}$, one obtains the following conditional expectation:  
\[\Ex\left[x\mid x_t,\Ff_{t,i-1},\Ff_{t-1},\Yy_{t,i-1}\right] = x_t + B\overline{Y},\]
where $\overline{Y}$ is the average move for the samples from $x_t$, including a move of length 0 when the sample fails to find a new fiber element.
\begin{algorithm}[!b]
	\captionsetup{font=small}
	\caption{Parameter Update Loop (\textsc{Update})}\label{alg:ibfs}
	\begin{algorithmic}[1]		
		\Input $x_t\in \Ff,\, I\in \Z_{>0},\,\, J\in \Z_{>0},\, B,\,\alpha^+_0,\,\alpha^-_0,\,\beta^+_0,\,\beta^-_0,\,\Ff_{t-1}$
		\Output $\Ff_{t}$
		\State $\Ff_{t,0} = \Ff_{t-1}$
		\For{$i \in [I]$}
		\State $\{\alpha^+_{i},\,\alpha^-_{i},\,\beta^+_{i},\,\beta^-_{i},\,\Ff_{t,i}\} = \textsc{Sample}(x_t,\,J,\,B,\,\alpha^+_{i-1},\,\alpha^-_{i-1},\,\beta^+_{i-1},\,\beta^-_{i-1},\,\Ff_{t,i-1})$
		\EndFor
		\State $\Ff_t = \Ff_{t,I}$	
	\end{algorithmic}
\end{algorithm}
Since $Y_{i,j,k}$ is a mixture of Poisson random variables and $B$ spans $\ker_{\Z}A$, it follows that $\pr(BY_{i,j}=x-x_t) >0$ for all $x,x_t \in \Ff$. However, $Y_{i,j}$ with large or dense coefficients will have small probability relative to the probabilities corresponding to sparse coefficients with relatively small magnitudes that yield fiber elements that are close to the starting element $x_t$. Instead of increasing the number of samples $J$ or iterations $I$, whenever Algorithm \ref{alg:ibfs} begins to sample a large number of previously sampled fiber elements, a new element $x_{t+1}$ may be chosen in order to sample from a different area of the fiber. 

The goal of this is to select a new starting element in such a way that the necessary magnitude and sparsity of moves to unsampled fiber elements is decreased, increasing their likelihood of being sampled. A simple heuristic strategy for selecting such a point is to uniformly sample the next starting element from $\Ff_t^* = \Ff_t\backslash \Ff_{t-1}$. In other words, randomly select the next starting element from the set of all newly-sampled fiber elements. Algorithm \ref{alg:mibfs} illustrates this process for a mixture of uniform samples from $\Ff_t^*$ and $\Ff_t$:
\begin{equation}
	\label{eqn:select}
	\pr^{(t)}(x\mid \Ff_t, \Ff_{t-1}) = \pi\dunif(\Ff_t^*) + (1-\pi)\dunif(\Ff_t),
\end{equation}
where $\pi \in [0,1]$ and $\pi = 1$ whenever $\Ff_t^* = \emptyset$.
\begin{algorithm}[!h]
	\captionsetup{font=small}
	\caption{Random Updating Moving Bayesian Algorithm}\label{alg:mibfs}
	\begin{algorithmic}[1]		
		\Input $x_0\in \Ff,T\in \Z_{>0},I\in \Z_{>0},J\in \Z_{>0},B, \alpha^+_0,\,\alpha^-_0,\,\beta^+_0,\,\beta^-_0$
		\Output $\Ff_T$
		\State $\Ff_0 = \{x_0\}$
		\For{$t\in [T]$}
		\State $\Ff_t = \textsc{Update}(x_t,\, I,\,\, J,\, B,\,\alpha^+_0,\,\alpha^-_0,\,\beta^+_0,\,\beta^-_0,\,\Ff_{t-1})$
		\State $x_{t+1} \sim \pr^{(t)}(x\mid\Ff_t,\Ff_{t-1})$ \label{step:x_t+1}
		\EndFor			
	\end{algorithmic}
\end{algorithm}

\subsection{Convergence to the fiber}

The ultimate purpose of the RUMBA algorithm is fiber discovery, so it is necessary that the fiber sample generated by the algorithm converges in probability to the actual fiber in the runtime parameters: $J= \#(\text{samples})$,  $I= \#(\text{iterations})$, and $T= \#(\text{steps})$. Each of these three parameters correspond to Algorithms \ref{alg:bfs}, \ref{alg:ibfs}, and \ref{alg:mibfs}, respectively.
The following three results prove that the partially discovered subsets of the fiber  obtained in each iteration and time step converge to the fiber as sample size, number of iterates, and time steps grow. 

Recall that $\Ff:=\fiber$  denotes the full fiber. 
Algorithm \ref{alg:bfs}, SAMPLE, is the basic fiber sample loop. It outputs a  sample $\Ff_{t,i}$ of the fiber, for a fixed time step $t$ and fixed iteration $i$. 
\begin{thm}\label{J:conv}
	 $\Ff_{t,i} \to \Ff$ as the sample size $J \to \infty$ in Algorithm \ref{alg:bfs}.
\end{thm}
\begin{proof}
Let $t \in [T]$ and $i\in [I]$, denoting $\pr^{t,i}(X)= \pr(X\mid \Ff_{t,i-1} x_t,\alpha^+_{i-1},\,\alpha^-_{i-1},\,\beta^+_{i-1},\,\beta^-_{i-1})$ for some $x_t \in \Ff_{t-1}\subseteq \Ff$. Then for any $x\in \Ff$,
since each $Y_{i,j}$ for $j \in [J]$ is sampled independently given $\lambda_i^+=\alpha^+_{i-1}/\beta^+_{i-1}$ and $\lambda_i^-=\alpha^-_{i-1}/\beta^-_{i-1}$, it follows that 
\[\pr^{t,i}(x\notin \Ff_{t,i}\mid x\notin\Ff_{t,i-1}) = \prod_{j=1}^{J}\pr^{t,i}\left(x\ne x_t + BY_{i,j}\right) = \left[1-\pr^{t,i}\left(x= x_t + BY_{i,j}\right)\right]^J.\]
Now, the probability $\pr^{t,i}(x\notin \Ff_{t,i}\mid x\notin \Ff_{t,i-1})\to 0$ as $J \to \infty$, whenever $\pr^{t,i}\left(x= x_t + BY_{i,j}\right) >0$. Since $\colsp(B) = \ker_{\Z} A$ and $x-x_t\in \ker_{\Z} A$ there exists $y \in \Z^K$ such that $x - x_t = By$. This means
\[\pr^{t,i}\left(x= x_t + BY_{i,j}\right) = \sum_{\{y\in \Z^K\colon By = x-x_t\}}\pr^{t,i}(Y_{i,j} = y).\]
Since $\supp\left(Y_{i,j}\right)= \Z^K$, for these $y$'s $\pr^{t,i}(Y = y)>0$, making $\pr^{t,i}\left(x= x_t + BY_{i,j}\right) > 0$. 
\end{proof}

Algorithm \ref{alg:ibfs}, UPDATE, is the parameter update loop. It iteratively calls SAMPLE and outputs $\Ff_{t,I}$ for a fixed time step $t$. 
\begin{cor}\label{I:conv}
	$\Ff_{t,I} \to \Ff$ as the number of iterates $I \to \infty$ at step $t$, in Algorithm \ref{alg:ibfs}.
\end{cor}

\begin{proof}
	Since $|\Ff|<\infty$ and $\Ff_{t,I}\subseteq \Ff$, $\Ff_{t,I}$ must converge to some finite subset of the fiber. Then there is some $r\in [I]$ such that $\Ff_{t,i}^* = \emptyset$ for all integers $i  \in \{r,\dots,I\}$. Since $\lambda_i^{\pm}$ is only updated when new points are added to the fiber sample, $\lambda_i^{\pm}=\lambda_r^{\pm}$ for all such $i$. Since $Y_{i,j,k}^{\pm} \sim \dpois(\lambda_{i}^{\pm})=\dpois(\lambda_{r}^{\pm})$. Then as $I \to \infty$, the tail iterations $i$ all sample from the same distribution $\pr^{t,r}(Y=y)$. This is equivalent to letting $J\to \infty$ in Algorithm \ref{alg:bfs}, so from Theorem \ref{J:conv}, $\Ff_{t,I}$ must converge to $\Ff$.
\end{proof}

 RUMBA calls the previous algorithms for a predetermined number of  time steps. At each time step, it outputs the sample $\Ff_{T}$. 
\begin{thm}\label{T:conv}
	$\Ff_{T} \to \Ff$ as the number of steps $T \to \infty$ in Algorithm \ref{alg:mibfs} using a selection method with distribution (\ref{eqn:select}), where $\pi = \mathbf{1}_{\{\Ff_{t-1}^*\ne\emptyset\}}$ is an indicator random variable in $\pr^{(t)}(x\mid \Ff_t, \Ff_{t-1})$.
\end{thm}

\begin{proof}
	Since $\Ff$ is finite, there exists $\tau\in \Z_{>0}$ such that for all $t> \tau$, $\Ff^*_t = \emptyset$; that is, $\Ff_t = \Ff^*$ for all $t>\tau$ where $\Ff^* \subseteq \Ff$. Note that since $x_0\in \Ff_0$, $\Ff^* \ne \emptyset$. At each step $t$, 
	\[x_t \sim \pr^{(t)}(x\mid \Ff_t, \Ff_{t-1}) = \mathbf{1}_{\{\Ff_{t-1}^*\ne\emptyset\}}\dunif(\Ff^*_{t-1})+ \left(\mathbf{1}_{\{\Ff_{t-1}^*=\emptyset\}}\right)\dunif(\Ff_t)\]
	and for $t>\tau$
	\[x_t \sim \dunif(\Ff^*).\]
	For $x\in \Ff\backslash \Ff^*$, we introduce the following shorthand notation: 
	\[p_{x,x_t} = \sum_{y\in \Z^K\colon By= x-x_t}\pr^{t,0}(Y=y).\]
	For any $\lambda^{\pm}_0$, $p_{x,x_t} \in (0,1)$, such that the probability $x$ will not be the first sampled element, in the first iteration of Algorithm \ref{alg:ibfs}, for any $t$ where $\tau < t \le T$, as $T\to \infty$, is $\prod_{t=\tau+1}^\infty(1-p_{x,x_t})$. This diverges to 0 if $\sum_{t=\tau+1}^T \log(1-p_{x,x_t}) \to -\infty$, and since there are finitely many $x_t\in \Ff^*$, this sum is
	\[\sum_{t=\tau+1}^T \log(1-p_{x,x_t}) = \sum_{x_t\in \Ff^*}c_{T,t}\log(1-p_{x,x_t}),\]
	where $c_{T,t}$ is the number of times $x_t$ was sampled as the starting element over all $T$ steps. Since $x_t\sim \dunif(\Ff^*)$, there is some $t$ such that $c_{T,t}\to \infty$ as $T\to \infty$. Then 
	$\sum_{x_t\in \Ff^*}c_{T,t}\log(1-p_{x,x_t})\to -\infty$, such that $\prod_{t=\tau+1}^\infty(1-p_{x,x_t})$ diverges to zero, and the probability that $x$ will be sampled as a starting element converges to 1. This contradicts that $x\in \Ff\backslash\Ff^*$ since only elements in $\Ff^*$ can be sampled as starting elements. In other words, the sampled fibers converge in probability  $\Ff_T \to \Ff$  as $T\to \infty$.
\end{proof}

\section{Simulations and experiments}\label{section:simulations}

The code for the RUMBA sampler, Algorithm \ref{alg:mibfs}, and the examples contained in this section is located at the following GitHub repository: \url{https://github.com/mbakenhus/rumba_sampler}. All simulations were written in R version 4.1.2, and run on a Windows 11 laptop using Windows Subsystem for Linux (WSL) with Ubuntu release 22.04 LTS. The laptop hardware included an AMD Ryzen 9 5900HX CPU and 16 GB of RAM, with an 8 GB limit for WSL. Runtimes and the total number of discovered points for each of the simulations are given by the table in Fig. \ref{fig:runtime}.
{\footnotesize
	\begin{figure}[!h]
		\begin{tabular}{l|rrr|r|r}
			\textbf{Fiber} & $J$ & $I$ & $T$ & \textbf{Runtime} & \textbf{Elements}	\\
			\hline
			DS98  &  100 &   5 &   5 &    2.151s  & 1622    \\
			DS98  &  100 &  10 &   5 & 	  4.132s  & 2879    \\
			DS98  &  100 &   5 &  10 &    4.224s  & 3233    \\
			DS98  & 1000 &  15 & 200 & 2872.722s  & 1262912 \\
			\hline
			$5\times 5\times 5$    ($S=   1$) &  100 &  25 &  50 & 125.388s & 81934 \\
			$5\times 5\times 5$    ($S=0.65$) &  100 &  25 &  50 & 129.928s & 59752 \\
			$5\times 5\times 5$    ($S=0.35$) &  100 &  25 &  50 & 126.528s & 15711 \\
			\hline
			$10\times 10\times 10$ ($S=   1$) &  100 &  25 &  50 & 269.855s & 93471 \\
			$10\times 10\times 10$ ($S=0.65$) &  100 &  25 &  50 & 193.664s & 26281 \\
			$10\times 10\times 10$ ($S=0.35$) &  100 &  25 &  50 & 177.717s & 2716 \\
			\hline
			Single $A_k$ & 1000 & 8 &  32 & 211.472 & 2047 \\
			Split  $A_k$ & 1000 & 8 &  24 & 158.447 & 2041 \\
		\end{tabular}
		\caption{The table shows the total runtime and number of unique elements for each fiber according to the runtime parameters: $J$ samples , $I$ iterations, and $T$ steps.}\label{fig:runtime}
	\end{figure}
}

\subsection{Independence model}\label{section: ds98 table example}

\begin{figure}[!b]
	\includegraphics[scale=0.42]{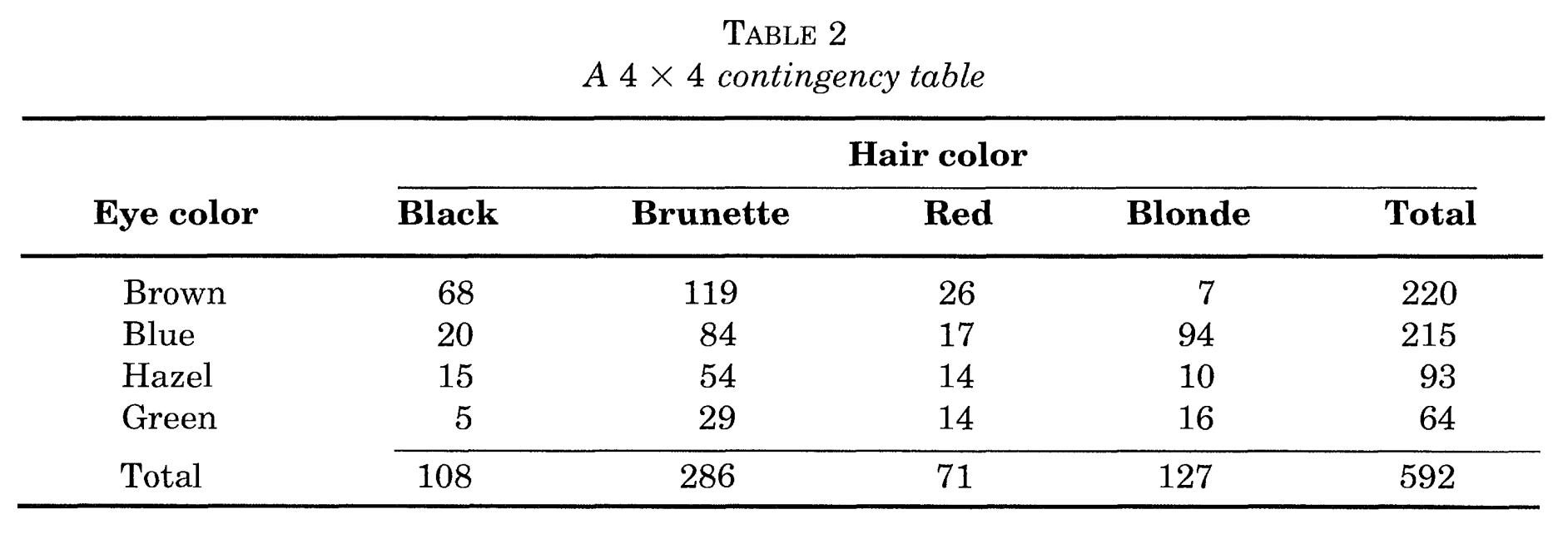}
	\caption{Table 2 from \cite{DS98} corresponding to the initial point $x_0$ in the fiber. Fiber sample results for this data table are  in Figures \ref{fig:DS98_compare} and \ref{fig:DS98_large}.}\label{fig:DS98_table}
\end{figure}

\begin{figure}[!b]
	\includegraphics[scale=0.42]{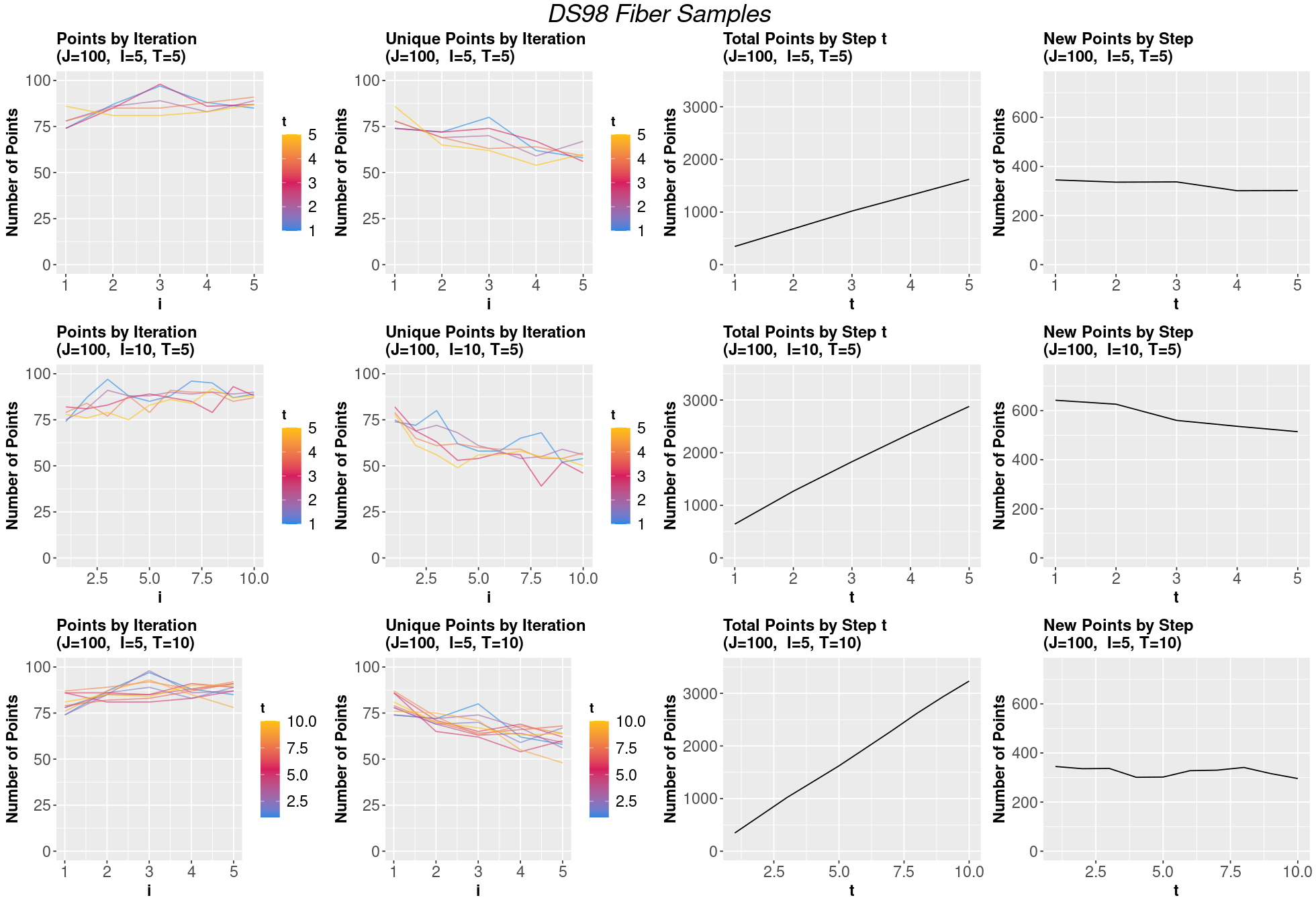}
	\caption{An illustration of Algorithm \ref{alg:mibfs} using the fiber corresponding to Table \ref{fig:DS98_table} using different parameter sets. From left to right columns, the plots describe: the total number of samples in the fiber at each iteration $i\in [I]$ for each step $t\in [T]$; the total number of samples that found new points at each iteration $i\in [I]$ and step $t\in[T]$; the cumulative total of unique fiber points sampled at each step $t\in [T]$; the number of new points sampled at each step $t\in[T]$.}\label{fig:DS98_compare}
\end{figure}

\begin{figure}[!h]
	\includegraphics[scale=0.42]{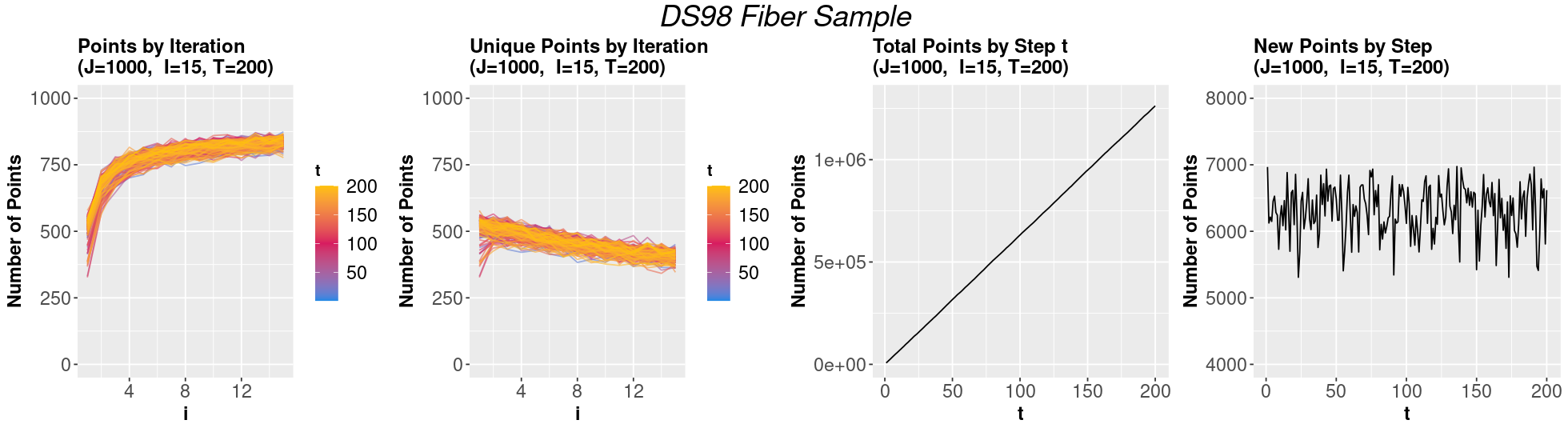}
	\caption{Algorithm \ref{alg:mibfs} sampling the fiber of the independence model on a $4\times4$ contingency table using the data point from Figure \ref{fig:DS98_table}. The algorithm was run with  $T=200$ steps, $I=15$ iterations per step, and $J=1000$ samples, for a total of 3,000,000 samples. These samples yielded 1,262,912 unique fiber elements.}\label{fig:DS98_large}
\end{figure}

We consider the independence model from \cite{DS98}, where the authors reported poor performance of the lattice bases Markov chain for exploring the fiber of the data table  presented in Figure \ref{fig:DS98_table}. 

As usual in algebraic statistics, the initial point $x_0$ corresponds to the data table by  flattening of this $4\times 4$ table, namely, $x_0=(68,20,\dots,10,16)\in\mathbb Z^{16}$. The right-hand side $u$ corresponds to the sufficient statistics of the data table under the model of independence, which are the row and column sums: $u=(220,215,93,64,108,286,71,127)\in\mathbb Z^8$. The $16\times8$ matrix $A$ is the $0/1$ matrix corresponding to the linear map which computes row and column matrix sums. 
 Our simulations illustrate that  running the RUMBA sampler with biasing and updating Poisson parameters along the run does sample the fiber at a near-constant rate.

In Figure~\ref{fig:DS98_compare}, each row of the figure corresponds to one set of values of the runtime parameters $I$, $J$, and $K$. The plots in the first  column show the number of sampled points in $\fiber$ for each iteration and at each time step. Since $J=100$, each time 100 samples are generated, between 75 and 100 of the points land in the fiber, and the rest are outside and therefore rejected. After a few samples, it is expected that the algorithm will start seeing the same points in the fiber instead of discovering new ones. This is confirmed by the (slight) negative slope of the graphs in the second column, where we plot only the new fiber points discovered. As $I$ increases from $1$ to $10$, the sampler discovers less points each time, but never less than $50$. 
Noting that colors and axes are the same as in the first column, 
the fact that for each time step corresponding to each color of the graph we see similar performance  indicates that the moving Step~\ref{step:x_t+1} of Algorithm~\ref{alg:mibfs} is helping the sampler move along the fiber faster; cf. Step 4 in the algorithm outline on page~\pageref{algo outline}. The third column is very informative for the overall performance of the RUMBA sampler, as it depicts the \emph{cumulative total of unique fiber points} sampled at each step $t$. There is no indication of the sampler stopping to discover new fiber points in this very large fiber. The forth column shows us that the fiber discovery rate remains fairly constant across time steps $t$. 

Figure~\ref{fig:DS98_large}  illustrates longer algorithm performance. The four columns of the figure have the same meaning as before. We see that after   $T=200$ time steps, $I=15$ iterations per step, and $J=1000$ samples, for a total sample size of 3,000,000, RUMBA sampler discovered 1,262,912 unique fiber elements.

\subsection{Sparse $Q\times Q \times Q$ contingency tables}

\begin{figure}[!b]
	\includegraphics[scale=0.42]{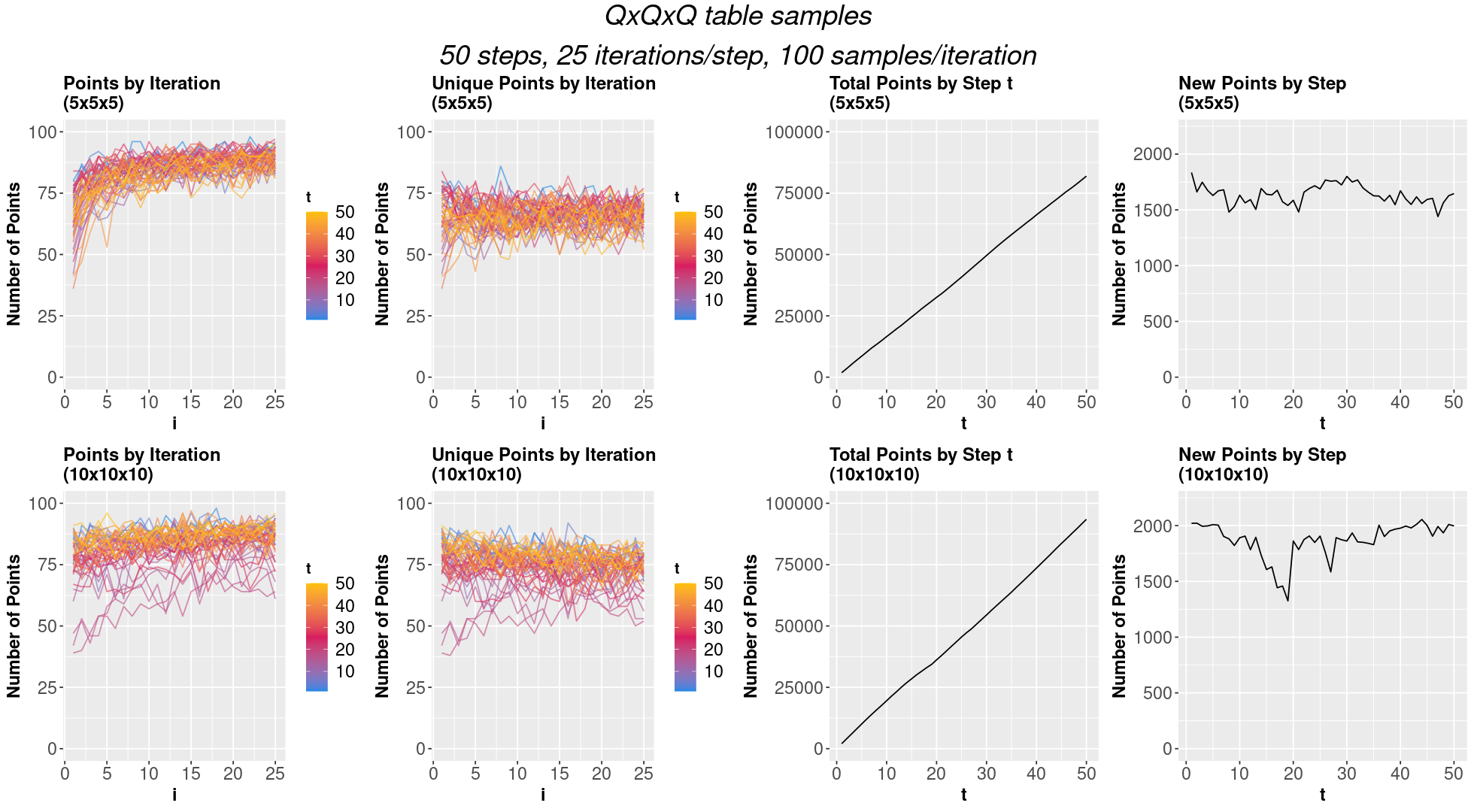}
	\caption{The results from sampling from the fiber using tables simulated from $\Uu_{Q,1}$. The top row shows results for $Q=5$ and the bottom row for $Q=10$.}\label{fig:QxQxQ}
\end{figure}

\begin{figure}[!b]
	\includegraphics[scale=0.42]{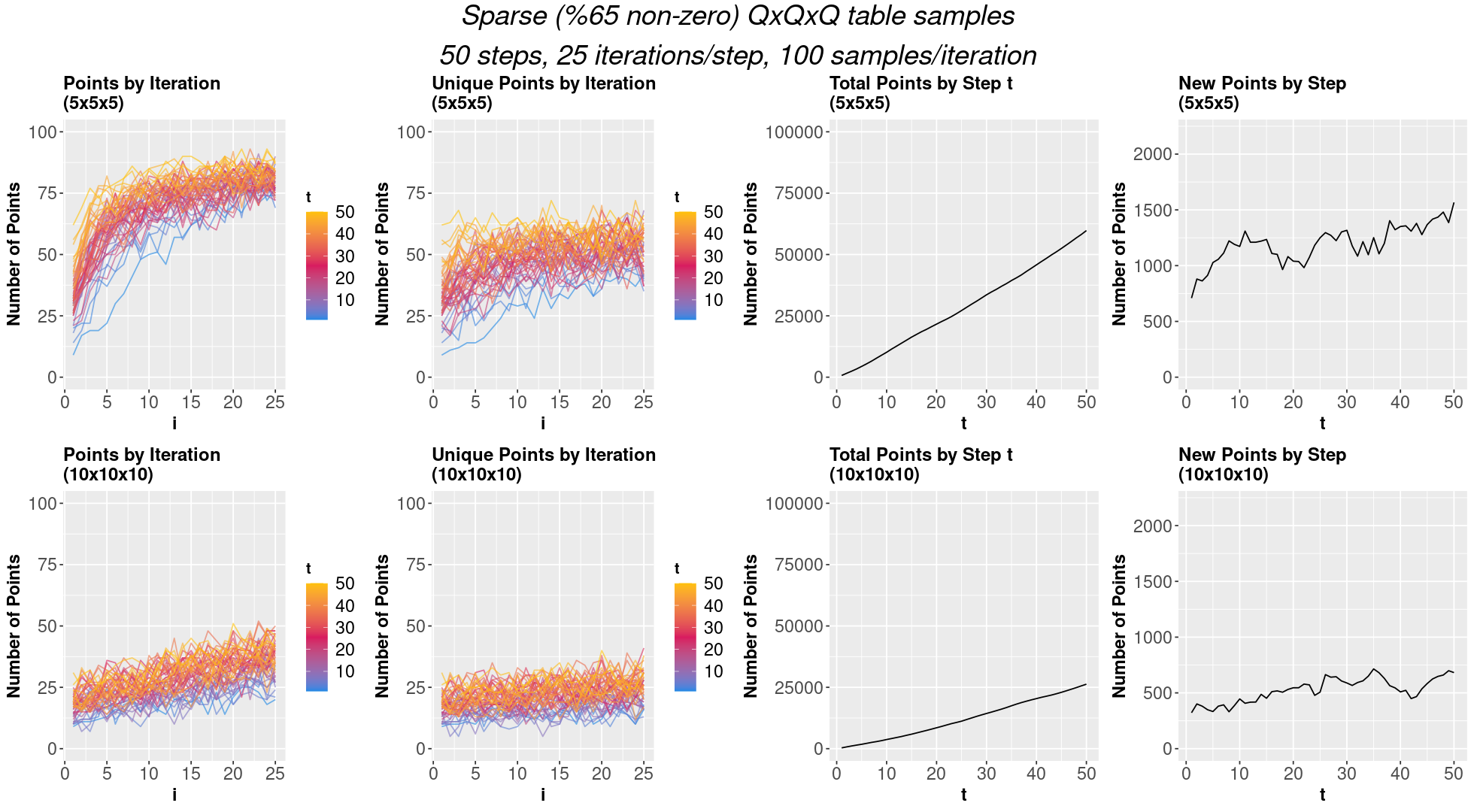}
	\caption{The results from sampling from the fiber using tables simulated from $\Uu_{Q,0.65}$. For sparse tables the sampler found fewer elements in the fiber at each iteration than for dense tables.}\label{fig:QxQxQ_65}
\end{figure}

\begin{figure}[!b]
	\includegraphics[scale=0.42]{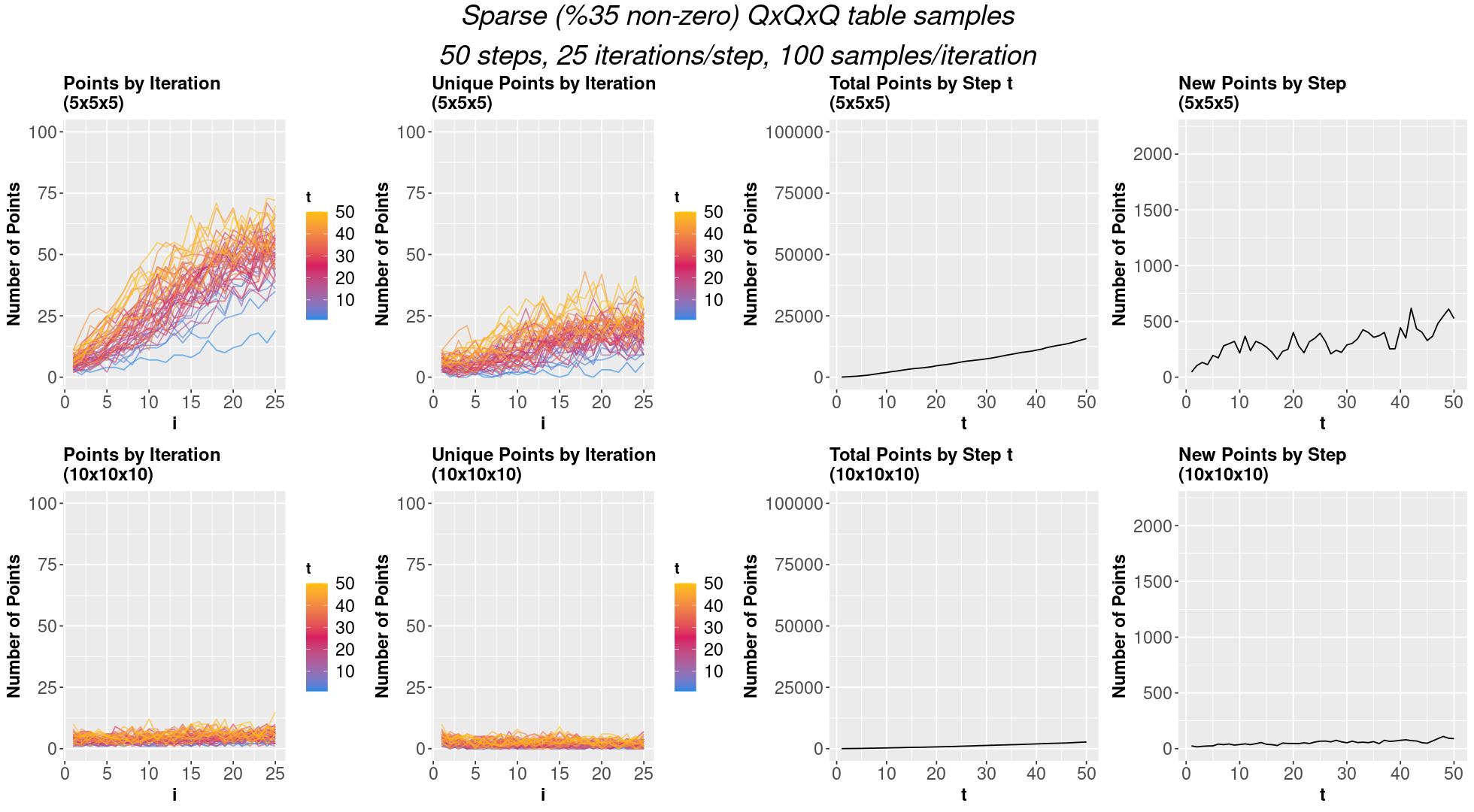}
	\caption{The results from sampling from the fiber using tables simulated from $\Uu_{Q,0.35}$. Sparse tables with high dimension did not yield large samples, and may require basis moves with better connectivity.}\label{fig:QxQxQ_35}
\end{figure} 

The independence model was the easiest configuration matrix $A$ to consider, because its Markov basis is known to be quadratic. In contrast, for no-three-factor interaction model on three-way contingency tables, depending on the levels of the three random variables, Markov bases can be either of bounded complexity or arbitrarily complicated; see \cite[Theorems 1.2.17 and 1.2.18]{DSS09} and also a summary in \cite[Section 3.2]{MarkovBases25years}. Even so, for large enough matrices, sparsity structure in the data, or the initial starting point $x_0$, can complicate the behavior of sampling algorithms because many proposed moves can be rejected.  

We sampled the fiber for a no-three-factor interaction model on $Q\times Q \times Q$ contingency tables, using simulated data at different sparsity levels. We considered $Q = 5$ and $Q= 10$. 
 \cite{HAT12} give the configuration matrix $A$  for this model as:
\def\matriximgA{%
	\begin{matrix*}[r]
		A^*     &        &     \\
		& \ddots &     \\  
		&        &   A^* \\ 
		I_{Q^2} &  \dots & I_{Q^2}  
	\end{matrix*}
}%
\[
A=\Lambda^{(Q)}(M)=\left(\vphantom{\matriximgA}\right.\kern-2\nulldelimiterspace
\overbrace{\matriximgA}^{Q \text{ times}}\kern-\nulldelimiterspace\left.\vphantom{\matriximgA}\right)\in\mathbb Z^{3Q^2 \times Q^3},
\]

where $I_{Q^2}$ is the $Q^2\times Q^2$ identity matrix and $A^*\in\Z^{2Q\times Q^2}$ is the following two-factor configuration matrix for $Q\times Q$ tables:
\def\matriximgM{%
	\begin{matrix*}[r]
		\bm{1}_Q^T &        &            \\
		           & \ddots &            \\
		           &        & \bm{1}_Q^T \\
		       I_Q &  \dots &        I_Q 
	\end{matrix*}
}%
\[
A* =\left(\vphantom{\matriximgM}\right.\kern-2\nulldelimiterspace
\overbrace{\matriximgM}^{Q \text{ times}}\kern-\nulldelimiterspace\left.\vphantom{\matriximgM}\right). 
\]

Additionally, \cite{HAT12} show how to generate the lattice basis $B\in \Z^{Q^2 \times (Q-1)K}$ for $\Lambda^{(Q)}(A^*)$ when the lattice basis for $A^*$ is $B^*\in \Z^{Q^2 \times K}$:
\def\matriximgB{%
\begin{matrix*}[r]
           B^* &         &               \\
             &  \ddots &               \\
             &         &  B^* \\
          -B^* &   \dots & -B^*  
\end{matrix*}
}%
\[
B=\left(\vphantom{\matriximgB}\right.\kern-2\nulldelimiterspace
\overbrace{\matriximgB}^{Q-1 \text{ times}}\kern-\nulldelimiterspace\left.\vphantom{\matriximgB}\right). 
\]

Initial points $x_0\in\mathbb Z^{Q^3}$ 
 in the fiber are vectorized $Q\times Q\times Q$ tables generated using three different sparsity levels,  $S = 1$, $S= 0.65$, and $S= 0.35$. The idea for generating random sparse tables as starting points in the fiber follow the simulations done by \citeauthor{HAT12}.  Namely, the sparsity level $S$ is used to construct the support of the table,  $\Uu_{Q,S}$, defined to be  the set of cells in the table which are allowed to have non-zero entries; otherwise they are set to zero. 
  For example, when $S=0.65$, this means that $65\%$ of the entries of the vector $x_0$ are nonzero. 
 We construct $\Uu_{Q,S}$  by sampling $S Q^3$ (rounded to the nearest integer)  elements without replacement from $[Q]\times [Q] \times[Q]$. 
To populate nonzero entries of $x_0$,  $Q\times Q\times Q$ tables are simulated by sampling $n=5Q^3$ values from each $\Uu_{Q,S}$ with replacement. The simulated tables are flattened into a vector and as such used as an initial fiber element $x_0$  for Algorithm \ref{alg:mibfs}. 

For our simulations, we provided the sampler the lattice basis $B$, and set the initial distribution parameters to $\alpha_{0,k}^{\pm} = \frac{1}{K}$ and $\beta_{0,k}^{\pm}=1$ for $k\in[K]$ where $K$ is the number of basis moves. Runtime parameters were $T=50$ steps, $I =25$ iterations per step, and $N=100$ samples per iteration. Figures \ref{fig:QxQxQ}, \ref{fig:QxQxQ_65}, and \ref{fig:QxQxQ_35} show the results for $\Uu_{Q,1}$, $\Uu_{Q,0.65}$ and $\Uu_{Q,0.35}$, respectively. The columns in each figure represent the same quantities as in the previous section.  Simulations on dense tables do not show a slow down in the $10\times10\times10$ example compared to $5\times5\times5$, with the final sample actually being smaller for the latter than the former.
On the other hand, table sparsity appears to negatively affect the efficiency of Algorithm \ref{alg:mibfs}. There are several possibilities for why this might be the case. Sparse tables may generate fibers that contain fewer elements than dense tables. In this case, the lattice bases require dense combinations in order to connect fiber elements. Similar issues arose for the fiber in Section \ref{section: HW fibers}. Methods for dealing with this issue are discussed further in Section \ref{section:practical}.

\subsection{A family of matrices with segmented fibers}\label{section: HW fibers}

\begin{figure}[!b]
	\includegraphics[scale=0.42]{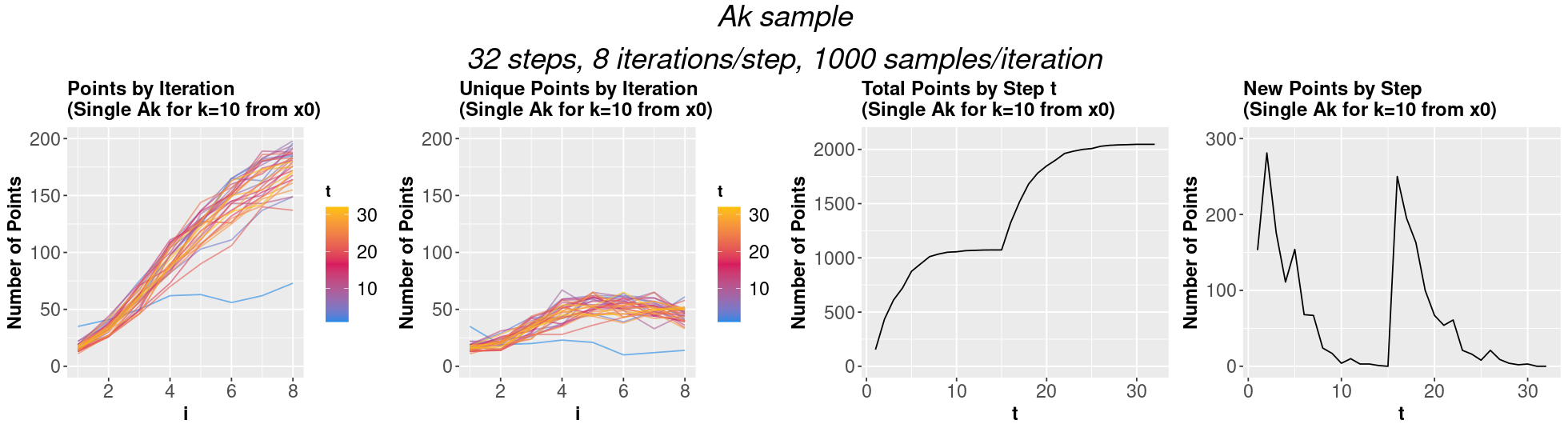}
	\caption{Results from sampling the fiber for $A_k$ with $u= \bm{e}_{2k+1}$, from the initial solution $x_0$.}\label{fig:Ak_single}
\end{figure}

\begin{figure}[!b]
	\includegraphics[scale=0.42]{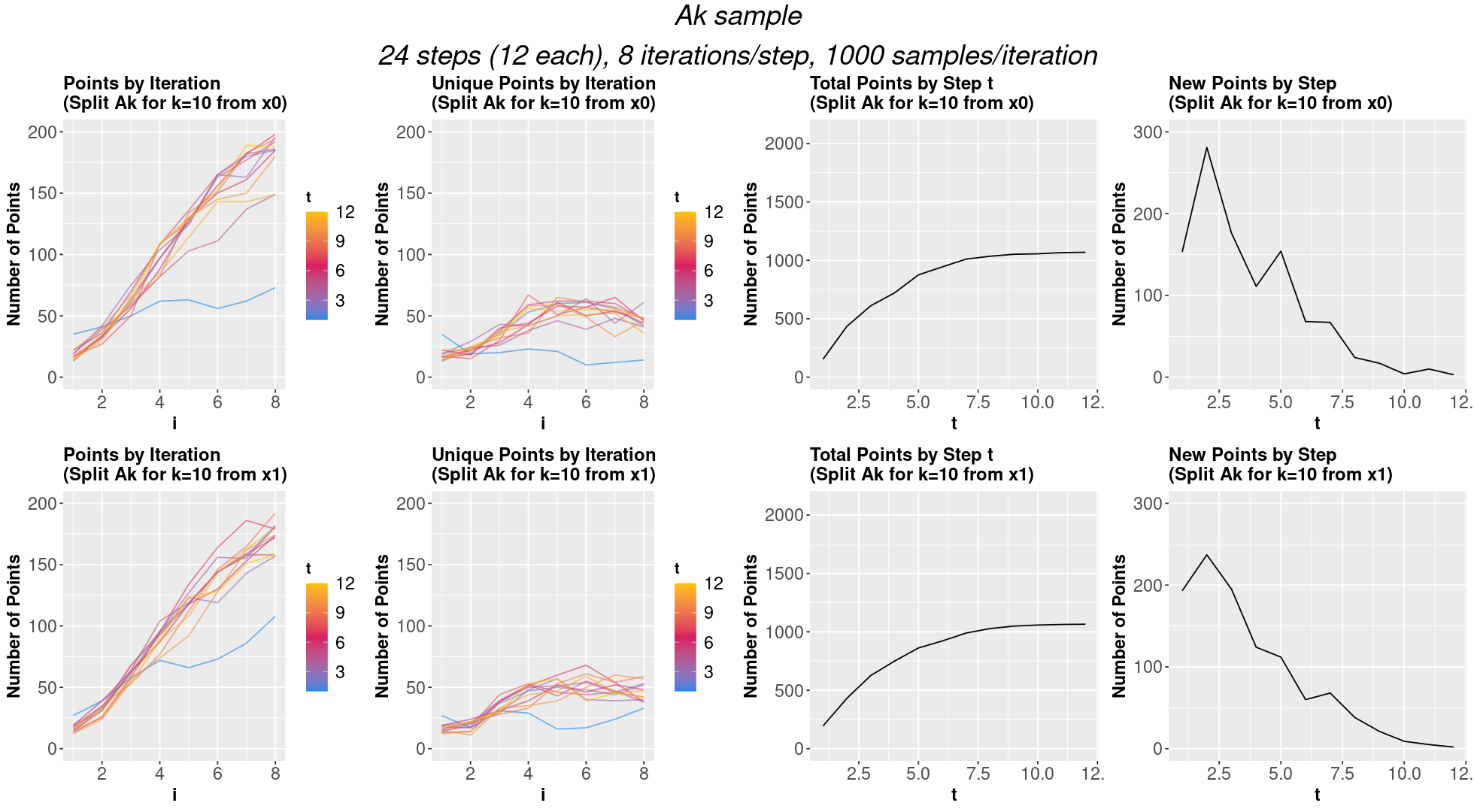}
	\caption{Results from sampling each segment for $A_k$ individually, using the the initial solutions $x_0$ and $x_1$.}\label{fig:Ak_split}
\end{figure}

Next, we turn to fibers which are nearly impossible to sample with Markov or even Graver bases. We take the following family of examples from \cite{HemmeckeWindisch:FiberGraphs}. 
For an integer $k$, the configuration matrix $A$ is: 
\[A= A_k = \begin{pmatrix} 
	I_k & I_k & {\bf 0} &{\bf 0}    & {\bf -1}_k & {\bf 0} \\      
	 {\bf 0}& {\bf 0}&I_k&I_k& {\bf 0}& {\bf -1}_k\\
	  {\bf 0}& {\bf 0}& {\bf 0}& {\bf 0}&  {\bf 1}_k& {\bf 1}_k
	 \end{pmatrix}\in\mathbb Z^{(2k+1)\times(4k+2)},\]
where $I_k$ is the $k\times k$ identity matrix, and ${\bf 1}_k\in\mathbb Z^k$ is the vector of all ones. 
\citeauthor{HemmeckeWindisch:FiberGraphs} use  $A_k$  in conjunction with a specially selected $u$ to prove that Markov chains using any of these algebra basis of the matrix $A$ will have poor mixing times due to the particular fiber graph structure. In essence, they construct a right-hand-side $u$ such that fiber graphs $G(\fiber,B)$ for various bases $B$ had low connectivity similar to the graph in Figure~\ref{figure: bad fibers}. In particular, when $u = \bm{e}_{2k+1}$, for $A_k x = u$, $\fiber$ lies in a high-dimensional space such that the supports of the vectors in one segment do not intersect with the supports of vectors in the other segment. For this $u$, the last two elements of $x$ must be one of $(x_{4k+1}, x_{4k+2}) = (1,0)$ or $(x_{4k+1}, x_{4k+2}) = (0,1)$, such that in the former case, $(x_1,\dots,x_k)\in\{0,1\}^k$ and $(x_{k+1},\dots,x_{2k})\in\{0,1\}^k$ are binary complements, while all other elements are 0. Similarly, in the latter case, $(x_{2k+1},\dots,x_{3k})\in\{0,1\}^k$ and $(x_{3k+1},\dots,x_{4k})\in\{0,1\}^k$ are binary complements, with the remaining elements 0 giving a total fiber size of $2^{k+1}$.

This type of a fiber structure requires a basis move that traverses between the two segments. \cite{HemmeckeWindisch:FiberGraphs} give this move as:
\[b^*=(\bm{0}_k, -\bm{1}_k,\bm{0}_k,\bm{1}_k, \bm{0}_k, -1,1)^{\tr},\]
which corresponds to the edge $G(\fiber, B)$ between the following two points: 
\[x_0=(\bm{0}_k, \bm{1}_k, \bm{0}_k, \bm{0}_k, \bm{0}_k, 1,0)^{\tr}\]
and 
\[x_1=(\bm{0}_k, \bm{0}_k, \bm{0}_k, \bm{1}_k, \bm{0}_k, 0,1)^{\tr}.\]

We ran RUMBA on this fiber for several values of $k$; here we show simulation results for $k=10$. In this case, it is known that the fiber size is $2048$. Any Markov chain using basis elements $B$ will have a low probability of selecting the one edge  in $G(\fiber,B)$ corresponding to the move $b^*$, because this edge must be selected exactly when the chain is sitting at either $x_0$ or $x_1$.  Of course, RUMBA does `notice' the fiber structure, although it is not constructing a Markov chain. It does not have to be at $x_0$ or $x_1$ to select that one-edge move. Instead, once the sampler finds elements in both segments, these elements may be used as the initial solution in subsequent steps of Algorithm \ref{alg:mibfs}. The likelihood of this occurring can be increased by taking $x_0$ or $x_1$ as the first initial solution, and setting the distribution parameters corresponding to $b^*$ to bias in favor of selecting this move. In our implementation, RUMBA selects the next initial solution from the most recent elements (if they exist) and then samples solutions from this point. The results from this method are pictured in Figure~\ref{fig:Ak_single}, where $x_0$ is the first initial solution. It can be seen in these plots that the number of new points found at each step flattens out, before jumping suddenly. This indicates that the sampler moved from sampling points in one segment to the other.

This method is relatively agnostic with respect to the fiber structure, only using it in the selection of the first initial solutions and its initial parameters. Figure \ref{fig:Ak_split} and the runtimes in \ref{fig:runtime}, illustrate how this may not be the most effective method. In this second simulation on the same fiber, RUMBA was run in sequence, first sampling using the starting solution $x_0$ then $x_1$, with the final sample taken as the union of the two runs. For $x_0$ and $x_1$, 1069 and 1065 fiber elements were discovered respectively, with 2041 unique elements discovered in total. As each segment contains 1028 elements, the sampler for both starting elements was able to sample elements in both segments during the first step, however as subsequent steps selected initial solutions from only the newly discovered points, the sampler picked a single segment and continued to sample elements from it until it no longer found new points. An unfortunate side effect of this is that once all points have been sampled from a segment, the selection of initial solutions is heavily biased towards picking points in this segment because there is a higher proportion of them in the full sample.  This seems to indicate that a more sophisticated selection of the initial solutions may be necessary for full discovery of fibers with complicated structures.


\section{Practical considerations for parameter tuning} \label{section:practical}


\subsection{Initial distribution parameters}

Algorithm \ref{alg:bfs} is sensitive to initial values $\alpha_0^{\pm}$ and $\beta_0^{\pm}$. Since $\beta$ parameters correspond to the cumulative sample size across iteration, these should be set such that 
\[\beta_{0,k}^{+} = \beta_{0,k}^{-} = 1, \text{ for all }k \in [K],\] 
because the initial sample is $\{x_t\}$ for each step. Any necessary tuning can then be performed by adjusting the $\alpha$ parameters, since
\[\lambda_{0,k}^{\pm} = \frac{\alpha_{0,k}^{\pm}}{\beta_{0,k}^{\pm}}.\]
Given $\beta_0^{\pm}=1$, in most cases, the choice of 
\[\alpha_{0,k}^{+} = \alpha_{0,k}^{-} = \frac{1}{K}, \text{ for all }k \in [K] \]
works well, since this will give the initial coefficient distribution of $Y_{i,j,k}^{\pm}\sim\dpois(1/K)$. The expected samples for this distribution are sparse and the variance is low for each coefficient, producing short jumps to nearby lattice points, thereby allowing for a local exploration of the fiber. 

Since each $x$ requires a total of $2K$ samples from a Poisson distribution, one for each move in the basis, the runtime may be greatly affected depending on how these values are sampled. Asymptotically, we assume that this algorithm is $\Oo(1)$, however, in practice for certain values of $\lambda$ it may be the case that an algorithm for sampling $\dpois(\lambda)$ with asymptotic complexity $\Oo(\lambda^{1/2})$ performs better than the $\Oo(1)$ algorithm. Even when the individual time difference between Poisson implementation may seem negligible, it is compounded by a factor of $T\times I\times J\times K$, after accounting for steps, iterations, samples, and basis size in Algorithm \ref{alg:mibfs}. While most library implementations of Poisson generation take parameter values into consideration, for large bases it may be prudent to check  each $\lambda^{\pm}_{i,k}$ as they change with each iteration of Algorithm \ref{alg:bfs}, and select a Poisson generation algorithm according to their values. For a comparison of different Poisson sampling algorithms, and which values of $\lambda$ they should be used for, see \cite{KEMP1990133} and \cite{kemp1991:poisson}.

\subsection{Runtime parameters}\label{section: runtime parameters}

Theorems \ref{J:conv}  and \ref{T:conv} and Corollary \ref{I:conv} imply that number of samples $J$, the number of parameter iterations $I$ and the number of steps $T$ should generally be set as large as is tolerable. That said, attention should be paid to the different effects of changing each of these three values. Setting $J$ too small will sometimes lead to the sampler failing to find any fiber elements, especially when the basis does not connect the fiber. Increasing $J$ increases the likelihood of observing extreme values for $Y_{i,j,k}^{\pm}$ in the samples, which corresponds to $Y_{i,j}$'s that are less sparse and have greater magnitude. This may be preferable when using a basis that does not connect the fiber, or when the fiber is given by a polytope of large diameter. However, if the initial $x_t$ is near the boundary of the polytope, increasing $J$ may lead to Algorithm \ref{alg:bfs} discarding a large number of initial samples that fall outside of the fiber. This should be avoided since it is essentially performing unnecessary operations that do not affect any of the parameters. One option to mitigate this is start with a smaller value of $J$ and increase it for later iterations when the $\lambda^{\pm}_{i,k}$'s are more biased toward sampling within the fiber. 

The number of parameter iterations $I$ should be relatively small when $J$ has been set appropriately, since updates to the parameters only occur when new fiber elements are discovered. Typically, Algorithm \ref{alg:ibfs} discovers elements that require sparse $Y_{i,j}$ relatively quickly and then biases new samples toward these elements. In practice this means the algorithm is more likely to re-sample previously discovered elements as $I$ increases. If the portion of the sample that is contained in the fiber at iteration $I$ consists mainly of fiber elements that have not been sampled in previous iterations, then $I$ should be increased. If this portion contains mostly previously sampled elements, decreasing $I$ may be preferable.

While the parameter iterations operate as a local discovery of fiber elements, the steps $t\in[T]$ are global in scope, when $x_t$ is sampled from the new set of discovered fiber points $\Ff_{t-1}^*$. In general $T$ should be as large as possible, especially for fibers with high diameter polytopes. Increasing $T$ allows for larger combinations of basis moves from the initial starting element $x_0$, without the need for sampling dense or high magnitude coefficient vectors $Y_{i,j}$. However, if $T$ is large and $I$ and $J$ are not, then fiber discovery becomes more dependent on the fiber connectivity for the given basis, since the sampled coefficient vectors will be relatively sparse. 

In general, for bases $\Bb$ that fully connect the fiber elements e.g. Markov, Gr\"{o}bner and Graver bases, sparse $Y_{i,j}$ are sufficient for sampling the fiber. Using such a basis does not require large values for the number of iterations $I$ since parameter iterations are primarily used to affect the sparsity. Instead emphasis can be placed on using a large number of steps and samples at each step. While bases with high fiber connectivity are preferable, they are not always computationally feasible. For a lattice basis that does not connect the fiber, care should be taken in the parameter tuning process since certain pairs $x, x_t \in \Ff$ will require dense combinations of basis elements to connect them. 


\section*{Acknowledgements} The authors are at the Department of Applied Mathematics, Illinois Institute of Technology. This work is supported by DOE/SC award \#1010629 and the Simons Foundation Travel Gift \#854770.
We are grateful for productive discussions and input from   Amirreza Eshraghi, F\'elix Almendra-Hern\'andez, David Kahle, Hisayuki Hara, and Liam Solus. In particular, we thank  Jes\'us A. De Loera for suggesting the name RUMBA for the sampler.

\bibliography{three-way-tables,AlgStatAndNtwks,MarkovBases,Algorithms,AlgebraAndRandomness}
\bibliographystyle{abbrvnat}

\section*{Appendix: notation} \label{section:appendix}

{\bf Indices:}
\begin{itemize}[label={}]
	\item $N$: \# Constraints
	\item $M$: \# Variables
	\item $K$: \# Basis Vectors
	\item $T$: \# Time Steps
	\item $I$: \# Iterations
	\item $J$: \# Samples
\end{itemize}
{\bf Matrices and Vectors:}
\begin{itemize}[label={}]
	\item $A\in \Z^{N\times M}$
	\item $x_0 \in \Z_{\ge 0}^M$
	\item $x_t \in \Ff_{t-1}$.
	\item $u \in \Z^N$ such that $Ax_0 = u$
	\item $B = \left[b_1\,\, b_2\,\,\dots\,\, b_K\right] \in \Z^{M\times K}$ such that $\colsp(B) = \ker_{\Z} A$
\end{itemize}
{\bf Random Variables:}
\begin{itemize}[label={}]
	\item $Y_{i,j,k}^{+} \sim \dpois(\lambda_{i,k}^{+})$ for $i \in [I]$, $j \in [J]$, and $k \in [K]$
	\item $Y_{i,j,k}^{-} \sim \dpois(\lambda_{i,k}^{-})$ for $i \in [I]$, $j \in [J]$, and $k \in [K]$
	\item $Y^{\pm}_{i,j} = \left(Y_{i,j,1}^{\pm},\dots,Y_{i,j,K}^{\pm}\right)$ for $i \in [I]$ and $j \in [J]$
	\item $Y_{i,j} = Y^+_{i,j} - Y^-_{i,j}$ for $i \in [I]$ and $j \in [J]$
	\item $X_{t,i,j} = x_t + BY_{i,j}$ for $t\in [T]$, $i\in [I]$
	\item $S_{i,j^*,k}^+ = \sum_{j \le j^*\colon X_{t,i,j}\in \Ff_{t,i}^*} Y_{i,j,k}^+$ for $i \in [I]$, $j^*\in[J]$, and $k \in [K]$
	\item $S_{i,j^*,k}^- = \sum_{j \le j^*\colon X_{t,i,j}\in \Ff_{t,i}^*} Y_{i,j,k}^-$ for $i \in [I]$, $j^*\in[J]$, and $k \in [K]$
	\item $S^{\pm}_{i,k} = S_{i,J,k}$ for $i \in [I]$ and $k \in [K]$
	\item $\Yy_{t,i}$ the set of $Y_{i,j}$ sampled at the $i^{\text{th}}$ iteration of the $t^{\text{th}}$ step
	\item $\Yy_{t,i}^* = \{Y_{i,j}\in \Yy_{t,i} \colon \exists j \in [J], \text{ such that } X_{t,i,j}\in \Ff_{t,i}^* \text{ and } X_{t,i,j} = x_t + BY_{i,j} \}$.
\end{itemize}
{\bf Samples:}
\begin{itemize}[label={}]
	\item $\Ff =\Ff_A(u)= \{x\in \Z_{\ge 0}^M\colon Ax = u \}$
	\item $\Ff_t$: Unique points sampled after $t \in [T]$ steps. 
	\item $\Ff_{0,0} = \{x_0\}$
	\item $\Ff_{t,0} = \Ff_{t-1}$ for $t \in [T]$
	\item $\Ff_{t,i}^*$: New points sampled a step $t\in [T]$ and iteration $i \in [I]$.
	\item $\Ff_{t,i} = \Ff_{t,i-1} \cup \Ff_{t,i}^*$ for $t \in [T]$
	\item $n_{t,i} = |\Ff_{t,i}^*|$
\end{itemize}
{\bf Iterated Parameters:}
\begin{itemize}[label={}]
	\item $\alpha^{+}_0 = \left(\alpha_{0,1}^{+},\dots, \alpha_{0,K}^{+}\right)\in \R^{K}_{\ge 0}$
	\item $\alpha^{-}_0 = \left(\alpha_{0,1}^{-},\dots, \alpha_{0,K}^{-}\right)\in \R^{K}_{\ge 0}$
	\item $\beta^{+}_0 =\beta^{-}_0 = \left(1,\dots, 1\right)\in \R^{K}_{> 0}$
	\item
	\item $\alpha_{i,k}^{\pm} = \alpha_{i-1,k}^{\pm} + S^{\pm}_{i,k}$ for  $i \in [I]\cup \{0\}$ and $k \in [K]$
	\item $\beta_{i,k}^{\pm} = \beta_{i-1,k}^{\pm} + n_{t,i}$ for  $i \in [I]\cup \{0\}$ and $k \in [K]$
	\item $\lambda^{\pm}_{i,k} = \frac{\alpha_{i,k}^{\pm}}{\beta_{i,k}^{\pm}}$ for $i \in [I]\cup \{0\}$ and $k \in [K]$
	\item
	\item $\alpha^{\pm}_i = \left(\alpha_{i,1}^{\pm},\dots, \alpha_{i,K}^{\pm}\right)\in \R^{m}_{\ge 0}$ for $i \in [I]\cup \{0\}$.
	\item $\beta^{\pm}_i = \left(\beta_{i,1}^{\pm},\dots, \beta_{i,K}^{\pm}\right)\in \R^{m}_{> 0}$ for $i \in [I]\cup \{0\}$.
	\item $\lambda^{\pm}_i = \left(\lambda_{i,1}^{\pm},\dots, \lambda_{i,K}^{\pm}\right)\in \R^{m}_{\ge 0}$ for $i \in [I]\cup \{0\}$.
	\item $\theta_i=(\alpha^+_i, \beta^+_i,\alpha^-_i,\beta^-_i)$ for $i \in [I]\cup \{0\}$.
\end{itemize}

\end{document}